\documentclass[runningheads]{llncs}

\usepackage{ellipsis, ragged2e} % bugfixing
\usepackage[l2tabu, orthodox]{nag} % avoid typical LaTeX errors

\usepackage[english]{babel}
\usepackage[T1]{fontenc}
\usepackage{lmodern}
\usepackage[utf8]{inputenc}
\usepackage[final]{microtype}

\usepackage{amsmath,amssymb}
\usepackage{csquotes}
\usepackage{booktabs}
\usepackage{enumitem}
\usepackage{framed}

\usepackage{pgfplots}
\pgfplotsset{compat=1.16}
\usepackage{tikz}
\usetikzlibrary{shapes,positioning,arrows,arrows.meta,calc,automata,matrix,fit,backgrounds}
\tikzstyle{state}+=[minimum size = 6mm, inner sep=0,outer sep=1]

\colorlet{disabled}{lightgray}
\tikzstyle{state}=[draw,rectangle,inner sep=5pt,rounded corners=2pt,minimum size=6mm]
\tikzstyle{action}=[font=\small,inner sep=0pt,outer sep=3pt]
\tikzstyle{actionnode}=[circle,draw=black,fill=black,minimum size=1mm,inner sep=0,outer sep=0]
\tikzstyle{actionedge}=[draw,-]
\tikzstyle{prob}=[font=\scriptsize,inner sep=0pt,outer sep=1pt]
\tikzstyle{probedge}=[draw,->]
\tikzstyle{directedge}=[draw,->]
\tikzset{chainarrow/.tip={Stealth[length=3pt]}}
\tikzset{>=chainarrow}

\usepackage{xparse}
\usepackage{mathtools}
\usepackage{environ}
\usepackage{marvosym}
\usepackage{bbm}
\usepackage{fontawesome5}

\usepackage{algorithmicx,algorithm}
\usepackage{algpseudocodex}

\usepackage[capitalize]{cleveref}
\DeclarePairedDelimiter{\delimabs}{\lvert}{\rvert}
\DeclarePairedDelimiter{\delimcardinality}{\lvert}{\rvert}
\DeclarePairedDelimiter{\delimnorm}{\lVert}{\rVert}

\NewDocumentCommand{\abs}{sm}{\IfBooleanTF{#1}{\delimabs*{#2}}{\delimabs{#2}}}
\NewDocumentCommand{\cardinality}{sm}{\IfBooleanTF{#1}{\delimcardinality*{#2}}{\delimcardinality{#2}}}
\NewDocumentCommand{\norm}{sm}{\IfBooleanTF{#1}{\delimnorm*{#2}}{\delimnorm{#2}}}
\NewDocumentCommand{\powerset}{r()}{2^{#1}}
\newcommand{\setcomplement}[1]{\overline{#1}}
\newcommand{\indicator}[1]{\mathbbm{1}_{#1}}

\newcommand{\unionSym}{\cup}\newcommand{\unionBin}{\mathbin{\unionSym}}\newcommand{\union}{\unionBin}

\newcommand{\intersectionSym}{\cap}\newcommand{\intersectionBin}{\mathbin{\intersectionSym}}\newcommand{\intersection}{\intersectionBin}
\newcommand{\UnionSym}{\bigcup}\newcommand{\Union}{\UnionSym}

\newcommand{\Naturals}{\mathbb{N}}

\newcommand{\Reals}{\mathbb{R}}

\DeclareMathOperator{\support}{supp}

\NewDocumentCommand{\Measures}{d()}{\IfValueTF{#1}{\Pi(#1)}{\Pi}}
\NewDocumentCommand{\Distributions}{d()}{\IfValueTF{#1}{\mathcal{D}(#1)}{\mathcal{D}}}
\NewDocumentCommand{\integral}{d<> m m}{\IfValueTF{#1}{\int_{#1} #2\,d#3}{\int #2\,d#3}}
\NewDocumentCommand{\Expectation}{s d[]}{\IfValueTF{#2}{\mathbb{E}}{\mathbb{E}\IfBooleanTF{#1}{\left[#2\right]}{[#2]}}}
\NewDocumentCommand{\Probability}{s d[]}{\mathop{\mathrm{Pr}}\IfValueT{#2}{\IfBooleanTF{#1}{\left[#2\right]}{[#2]}}}

\newcommand{\MC}{\mathsf{M}}

\newcommand{\States}{S}
\newcommand{\initialstate}{{\hat{s}}}
\NewDocumentCommand{\mctransitions}{d()}{\IfValueTF{#1}{\delta(#1)}{\delta}}

\NewDocumentCommand{\MCrestricted}{r<>}{\MC|_{#1}}
\NewDocumentCommand{\mctransitionsrestricted}{r<> d()}{\IfValueTF{#2}{\delta|_{#1}(#2)}{\delta|_{#1}}}

\newcommand{\reward}{r}
\DeclareMathOperator{\meanpayoff}{mp}

\newcommand{\infinitepath}{\rho}

\NewDocumentCommand{\Infinitepaths}{d<>}{\IfValueTF{#1}{\mathsf{Paths}_{#1}}{\mathsf{Paths}}}
\NewDocumentCommand{\Finitepaths}{d<>}{\IfValueTF{#1}{\mathsf{FPaths}_{#1}}{\mathsf{FPaths}}}
\DeclareMathOperator{\SccsOp}{SCC}\NewDocumentCommand{\Sccs}{r()}{\SccsOp(#1)}
\DeclareMathOperator{\BsccsOp}{BSCC}\NewDocumentCommand{\Bsccs}{r()}{\BsccsOp(#1)}
\NewDocumentCommand{\ProbabilityMC}{s r<> d[]}{\mathsf{Pr}_{#2}\IfNoValueF{#3}{\IfBooleanTF{#1}{\!\left[#3\right]\!}{[#3]}}}
\NewDocumentCommand{\ExpectationMC}{s r<> r[]}{\mathbb{E}_{#2}\IfBooleanTF{#1}{\!\left[#3\right]\!}{[#3]}}

\NewDocumentCommand{\ExpectedSum}{m m}{#1\langle#2\rangle}
\NewDocumentCommand{\ExpectedSumMC}{m m m}{\ExpectedSum{#1(#2)}{#3}}
\newcommand{\Reachset}{T}
\NewDocumentCommand{\stepreach}{r<>}{\Diamond^{=#1}}
\NewDocumentCommand{\boundedreach}{r<>}{\Diamond^{{\leq}#1}}
\newcommand{\reach}{\Diamond}

\NewDocumentCommand{\steadystate}{d<> d()}{\IfValueTF{#1}{\pi^\infty_{#1}}{\pi^\infty}\IfValueT{#2}{(#2)}}

\newtoggle{arxiv}
\toggletrue{arxiv}
\iftoggle{arxiv}{}{
\setlength{\textfloatsep}{0.5\textfloatsep} % Space below/above a float if its [t]
\setlength{\intextsep}{0.5\intextsep} % Space below/above a float if its [h]
\setlength{\floatsep}{0.5\floatsep} % Space between floats
\AtBeginDocument{
\setlength{\abovedisplayskip}{0.5\abovedisplayskip}
\setlength{\abovedisplayshortskip}{0.5\abovedisplayshortskip}
\setlength{\belowdisplayskip}{0.5\belowdisplayskip}
\setlength{\belowdisplayshortskip}{0.5\belowdisplayshortskip}
}

}

\usepackage[final]{todonotes}

\begin{document}
\title{Correct Approximation of Stationary Distributions}

\author{Tobias Meggendorfer\orcidID{0000-0002-1712-2165}}
\authorrunning{T.~Meggendorfer}

\institute{Institute of Science and Technology Austria\\
\email{tobias.meggendorfer@ista.ac.at}
}

\maketitle

\begin{abstract}
	A classical problem for Markov chains is determining their stationary (or steady-state) distribution.
	This problem has an equally classical solution based on eigenvectors and linear equation systems.
	However, this approach does not scale to large instances, and iterative solutions are desirable.
	It turns out that a naive approach, as used by current model checkers, may yield completely wrong results.
	We present a new approach, which utilizes recent advances in partial exploration and mean payoff computation to obtain a correct, converging approximation.
\end{abstract}

\section{Introduction}

\emph{Discrete-time Markov chains} (MCs) are an elegant and standard framework to describe stochastic processes, with a vast area of applications such as computer science \cite{DBLP:books/daglib/0020348}, biology \cite{paulsson2004summing}, epidemiology \cite{G_mez_2010}, and chemistry \cite{gillespie1976general}, to name a few.
In a nutshell, MC comprise a set of states and a transition function, assigning to each state a distribution over successors.
The system evolves by repeatedly drawing a successor state from the transition distribution of the current state.
This can, for example, model communication over a lossy channel, a queuing network, or populations of predator and prey which grow and interact randomly.
For many applications, the \emph{stationary distribution} of such a system is of particular interest.
Intuitively, this distribution describes in which states the system is in after an \enquote{infinite} number of steps.
For example, in a chemical reaction network this distribution could describe the equilibrium states of the mixture.

%Alternatively, there is a recent interest in viewing MC as \enquote{transformers} of distributions:
%Especially in chemical reactions, we can essentially treat the different substances as continuous, and interpret a distribution as their relative concentration.
%Then, in each step of the Markov chain, some part of the substances react, changing these relative concentrations.
%Here, the steady state distribution then describes the unique equilibrium \todo{asdasdf}

Traditionally, the stationary distribution is obtained by computing the dominant eigenvector for particular matrices and solving a series of linear equation systems.
This approach is appealing in theory, since it is polynomial in the size of the considered Markov chain.
Moreover, since linear algebra is an intensely studied field, many optimizations %and special cases
for the computations at hand are known.

In practice, these approaches however often turn out to be insufficient.
Real-world models may have millions of states, often ruling out exact solution approaches.
As such, the attention turns to iterative methods.
In particular, the popular model checker PRISM \cite{DBLP:conf/cav/KwiatkowskaNP11} employs the \emph{power method} (or \emph{power iteration}) to approximate the stationary distribution.
Similar to many other problems on Markov chains, such iterative methods have an exponential worst-case, however obtain good results quickly on many models.
(Models where iterative methods indeed converge slowly are called \emph{stiff}.)
However, as we show in this work, the \enquote{absolute change}-criterion used by PRISM to stop the iteration is incorrect.
In particular, the produced results may be arbitrarily wrong already on a model with only four states.
In \cite{DBLP:journals/tcs/HaddadM18,DBLP:conf/atva/BrazdilCCFKKPU14} the authors discuss a similar issue for the problem of \emph{reachability}, also rooted in an incorrect absolute change stopping criterion, and provide a solution through converging lower and \emph{upper} bounds.
In our case, the situations is more complicated.
The convergence of the power method is quite difficult to bound:
A good (and potentially tight) a-priori bound is given by the ratio of first and second eigenvalues, which however is as hard to determine as solving the problem itself.
In the case of MC, only a crude bound on this ratio can be obtained easily, which gives an exponential bound on the number of iterations required to achieve a given precision.
More strikingly, in contrast to reachability, there is to our knowledge no general \emph{adaptive} stopping criterion for power iteration, i.e.\ a way to check whether the current iterates are already close to the correct result.
Thus, one would always need to iterate for as many steps as given by the a-priori bound to obtain guarantees on the result.
In summary, exact solution approaches do not scale well, and the existing iterative approach may yield wrong results or requires an intractable number of steps.
%As such, a fundamentally different approach is required.

Another, orthogonal issue of the mentioned approaches is that they construct the \emph{complete} system, i.e.\ determine the stationary distribution for each state.
However, when we figure out that, for example, the stationary distribution has a value of at least $99\%$ for one state, all other states can have at most $1\%$ in total.
In case we are satisfied with an \emph{approximate} solution, we could already stop the computation here, without investigating any other state.
Inspired by the results of \cite{DBLP:conf/atva/BrazdilCCFKKPU14,DBLP:journals/lmcs/KretinskyM20}, we thus also want to find such an approximate solution, capable of identifying the relevant parts of the system and only constructing those.

\subsection{Contributions}

In this work, we address all the above issues.
To this end, we
\begin{itemize}
	\item
	provide a characterization of the stationary distribution through mean payoff which allows us to obtain provably correct approximations (\cref{sec:blocks}),
	\item
	introduce a general framework to approximate the stationary distribution in Markov chains, capable of utilizing partial exploration approaches (\cref{sec:framework}),
	\item
	as the main technical contribution, provide very general, precise correctness and termination proofs, requiring only minimal assumptions (\cref{stm:framework_correct}),
	\item
	instantiate this framework with both the classical solution approach as well as our novel sampling-based interval approximation approach (\cref{sec:framework:sampling}),
	\item
	evaluate the variants of our framework experimentally (\cref{sec:evaluation}), and
	\item
	demonstrate with a minimal example that the standard approach of PRISM may yield arbitrarily wrong results (\cref{fig:power_method_error}).
\end{itemize}

\subsection{Related Work}
Most related is the work of \cite{DBLP:conf/asmta/SpielerW13}, which also try to identify the most relevant parts of the system, however they employ the special structure given by cellular processes to find these regions and estimate the subsequent approximation error.
Many other works deal with special cases, such as queueing models \cite{DBLP:journals/questa/AdanFM09,DBLP:journals/orl/KimuraM21}, time-reversible chains \cite{DBLP:journals/mst/BressanPP20}, or positive rows (all states have a transition to one particular state) \cite{DBLP:journals/nla/BusicF11,DBLP:conf/asmta/FourneauQ12,DBLP:journals/amc/NesterovN15}.
In contrast, our methods aim to deal with general Markov chains.
We highlight that for the \enquote{positive row} case, \cite{DBLP:conf/asmta/FourneauQ12} also provides converging bounds, however through a different route.
Another topic of interest are continuous time Markov chains, where abstraction- and truncation-based algorithms are applicable \cite{DBLP:journals/siamrev/KuntzTSB21,DBLP:conf/qest/BackenkohlerBGW21} and computation of the stationary distribution can be used for time-bounded reachability \cite{DBLP:conf/qest/KatoenZ06}.

%https://link.springer.com/article/10.1007%2Fs11134-019-09599-x

\section{Preliminaries}

As usual, $\Naturals$ and $\Reals$ refer to the (positive) natural numbers and real numbers, respectively.
For a set $S$, $\setcomplement{S}$ denotes its complement, while $S^\star$ and $S^\omega$ refer to the set of finite and infinite sequences comprising elements of $S$, respectively.
We write $\indicator{S}(s) = 1$ if $s \in S$ and $0$ otherwise for the \emph{characteristic function} of $S$.

We assume familiarity with basic notions of probability theory, e.g., \emph{probability spaces}, \emph{probability measures}, and \emph{measurability}; see e.g.\ \cite{billingsley2008probability} for a general introduction.
A \emph{probability distribution} over a countable set $X$ is a mapping $d : X \to [0,1]$, such that $\sum_{x \in X} d(x) = 1$.
Its \emph{support} is denoted by $\support(d) = \{x \in X \mid d(x) > 0\}$.
$\Distributions(X)$ denotes the set of all probability distributions on $X$.
Some event happens \emph{almost surely} (a.s.) if it happens with probability $1$.

The central object of interest are Markov chains, a classical model for systems with stochastic behaviour:
	A (discrete-time time-homogeneous) \emph{Markov chain (MC)} is a tuple $\MC = (\States, \mctransitions)$, where
		$\States$ is a finite set of \emph{states}, and
		$\mctransitions : \States \to \Distributions(\States)$ is a \emph{transition function} that for each state $s$ yields a probability distribution over successor states.
We deliberately exclude the explicit definition of an initial state.
We direct the interested reader to, e.g., \cite[Sec.~10.1]{DBLP:books/daglib/0020348}, \cite[App.~A]{DBLP:books/wi/Puterman94}, or \cite{kulkarni2016modeling} for further information on Markov chains and related notions.

For ease of notation, %we overload functions mapping to distributions $f : Y \to \Distributions(X)$ by $f : Y \times X \to [0, 1]$, where $f(y, x) \coloneqq f(y)(x)$.
%For example, 
we write $\mctransitions(s, s')$ instead of $\mctransitions(s)(s')$, and, given %a distribution $d \in \Distributions(X)$ and 
a function $f : \States \to \Reals$ mapping states to real numbers, %elements of a set $X$ to real numbers, we write $\ExpectedSum{d}{f} \coloneqq \sum_{x \in X} d(x) f(x)$ to denote the weighted sum of $f$ with respect to $d$.
%For example, 
we write $\ExpectedSumMC{\mctransitions}{s}{f} \coloneqq \sum_{s' \in \States} \mctransitions(s, s') \cdot f(s')$ to denote the weighted sum of $f$ over the successors of $s$. %in an MC.

We always assume an arbitrary but fixed numbering of the states and identify a state with its respective number.
For example, given a vector $v \in \Reals^{\cardinality{\States}}$ and a state $s \in \States$, we may write $v[s]$ to denote the value associated with $s$ by $v$.
In this way, a function $v : \States \to \Reals$ is equivalent to a vector $v \in \Reals^{\cardinality{\States}}$.

For a set of states $R \subseteq S$ where no transitions leave $R$, i.e.\ $\mctransitions(s, s') = 0$ for all $s \in R$, $s' \in S \setminus R$, we define the \emph{restricted Markov chain} $\MCrestricted<R> \coloneqq (R, \mctransitionsrestricted<R>)$ with $\mctransitionsrestricted<R> : R \to \Distributions(R)$ copying the values of $\mctransitions$, i.e.\ $\mctransitionsrestricted<R>(s, s') = \mctransitions(s, s')$ for all $s, s' \in R$.

\paragraph*{Paths}
An \emph{infinite path} $\infinitepath$ in a Markov chain is an infinite sequence $\infinitepath = s_1 s_2 \cdots \in \States^\omega$, such that for every $i \in \Naturals$ we have that $\mctransitions(s_i, s_{i+1}) > 0$.
%A \emph{finite path} (or \emph{history}) $\finitepath = s_1 s_2 \dots s_n \in \States^\star$ is a non-empty, finite prefix of an infinite path of length $\cardinality{\finitepath} = n$, ending in some state $s_n$, denoted by $\last{\finitepath}$.
%Additionally, we define $\cardinality{\infinitepath} = \infty$ for infinite paths $\infinitepath$.
We use $\infinitepath(i)$ %and $\finitepath(i)$ 
to refer to the $i$-th state $s_i$ in a given infinite path.
%A state $s$ \emph{occurs} in an (in)finite path $\infinitepath$, denoted by $s \in \infinitepath$, if there exists an $i \leq \cardinality{\infinitepath}$ such that $s = \infinitepath(i)$.
We denote the set of all infinite paths of a Markov chain $\MC$ by %$\Finitepaths<\MC>$ 
$\Infinitepaths<\MC>$.
%Further, we use $\Finitepaths<\MC, s>$ ($\Infinitepaths<\MC, s>$) to refer to all (in)finite paths starting in state $s \in \States$.
Observe that in general %$\Finitepaths<\MC>$ and 
$\Infinitepaths<\MC>$ is a proper subset of %$\States^\star$ and 
$\States^\omega$, %, respectively,
as we imposed additional constraints.
%Moreover, $\Finitepaths<\MC>$ is countable in general, whereas $\Infinitepaths<\MC>$ is uncountable.
A Markov chain together with an initial state $\initialstate \in \States$ induces a unique probability measure $\ProbabilityMC<\MC, \initialstate>$ over infinite paths \cite[Sec.~10.1]{DBLP:books/daglib/0020348}.
Given a measurable random variable $f : \Infinitepaths<\MC> \to \Reals$, we write $\ExpectationMC<\MC, \initialstate>[f] \coloneqq \integral<\infinitepath \in \Infinitepaths>{f(\infinitepath)}{\ProbabilityMC<\MC, \initialstate>}$ to denote its expectation w.r.t.\ this measure.
%As such, it can also be viewed as a sequence of random variables $X_t$, where $X_{t+1}$ only depends on $X_t$ \cite[Appendix~A.1]{DBLP:books/wi/Puterman94}.

\paragraph*{Reachability}
An important tool in the following is the notion of \emph{reachability probability}, i.e.\ the probability that the system, starting from a state $\initialstate$, will eventually reach a given set $\Reachset$.
%\begin{definition}
	Formally, for a Markov chain $\MC$ and set of states $\Reachset$, we define the set of runs which reach $\Reachset$ (i)~at step $n$ by $\stepreach<n> \Reachset \coloneqq \{\infinitepath \in \Infinitepaths<\MC> \mid \infinitepath(n) \in \Reachset\}$ and % (ii)~within $n$ steps by $\boundedreach<n> \Reachset \coloneqq \Union_{i = 1}^n \stepreach<i> \Reachset$, and
	(ii)~eventually by $\reach \Reachset = \Union_{i = 1}^\infty \stepreach<i> \Reachset$.
	(For a measurability proof see e.g.\ \cite[Chp.~10]{DBLP:books/daglib/0020348}.)
%\end{definition}
For a state $\initialstate$, the probability to reach $\Reachset$ %at step $n$ and eventually, respectively, 
is given by $\ProbabilityMC<\MC, \initialstate>[\reach \Reachset]$.

Classically, the reachability probability can be determined by solving a linear equation system, as follows.
For a fixed target set $\Reachset$, let $S_0$ be all states that cannot reach $\Reachset$.
Note that $S_0$ can be determined by simple graph analysis.
Then, the reachability probability $\ProbabilityMC<\MC, \initialstate>[\reach \Reachset]$ is the unique solution of \cite[Thm.~10.19]{DBLP:books/daglib/0020348}
\begin{equation} \label{eq:reachability_value_fixpoint}
	f(s) = 1 \text{ if $s \in \Reachset$, } \quad 0 \text{ if $s \in S_0$, } \quad \text{ and } \quad \ExpectedSumMC{\mctransitions}{s}{f} \text{ otherwise}.
\end{equation}

\paragraph*{Value Iteration}
A classical tool to deal with Markov chains is \emph{value iteration} (VI) \cite{bellman1966dynamic}.
It is a simple yet surprisingly efficient and extendable approach to solve a variety of problems.
At its heart, VI relies, as the name suggests, on iteratively applying an operation to a value vector.
This operation often is called \enquote{Bellman backup} or \enquote{Bellman update}, usually derived from a fixed-point characterization of the problem at hand.
Thus, VI often can be viewed as fixed point iteration.
For reachability, inspired by \cref{eq:reachability_value_fixpoint}, we start from $v_1[s] = 0$ and iterate
\begin{equation} \label{eq:value_iteration}
	v_{k+1}[s] = 1 \text{ if $s \in \Reachset$, } \quad 0 \text{ if $s \in S_0$, } \quad \text{ and } \quad \ExpectedSumMC{\mctransitions}{s}{v_k} \text{ otherwise}.
\end{equation}
This iteration monotonically converges to the true value in the limit from below \cite[Thm.~10.15]{DBLP:books/daglib/0020348}, \cite[Thm.~7.2.12]{DBLP:books/wi/Puterman94}.
%Note that in this case, i.e.\ (unbounded) reachability, VI intrinsically can only yield approximate solutions.
Convergence up to a given precision may take exponential time \cite[Thm.~3]{DBLP:journals/tcs/HaddadM18}, but in practice VI often is much faster than methods based on equation solving.
For further details, see \iftoggle{arxiv}{\cref{app:value_iteration}}{\cite[App.~A.2]{arxiv}}.

\paragraph*{Strongly Connected Components}
A non-empty set of states $C \subseteq \States$ in a Markov chain is \emph{strongly connected} if for every pair $s, s' \in C$ there is a non-empty finite path from $s$ to $s'$.
Such a set $C$ is a \emph{strongly connected component} (SCC) if it is inclusion maximal, i.e.\ there exists no strongly connected $C'$ with $C \subsetneq C'$.
SCCs are disjoint, each state belongs to at most one SCC.
An SCC is \emph{bottom} (BSCC) if additionally no path leads out of it, i.e.\ for all $s \in C, s' \in S \setminus C$ we have $\mctransitions(s, s') = 0$.
The set of BSCCs in an MC $\MC$ is denoted by $\Bsccs(\MC)$ and can be determined in linear time by, e.g., Tarjan's algorithm~\cite{DBLP:journals/siamcomp/Tarjan72}. %w.r.t.\ the number of states and transitions by, e.g., Tarjan's algorithm~\cite{DBLP:journals/siamcomp/Tarjan72}.
%$\Sccs(\MC)$, $\Bsccs(\MC)$, and a topological ordering of the SCCs can be achieved in linear time w.r.t.\ the number of states and transitions by, e.g., Tarjan's algorithm~\cite{DBLP:journals/siamcomp/Tarjan72}.
%

The bottom components fully capture the limit behaviour of any Markov chain.
Intuitively, the following statement says that (i)~with probability one a run of a Markov chain eventually forever remains inside one single BSCC, and (ii)~inside a BSCC, all states are visited infinitely often with probability one.
\begin{lemma}[{\cite[Thm.~10.27]{DBLP:books/daglib/0020348}}] \label{stm:mc_bscc_properties}
	For any MC $\MC$ and state $s$, we have
	\begin{equation*}
		\ProbabilityMC<\MC, s>[\{\infinitepath \mid \exists R_i \in \Bsccs(\MC). \exists n_0 \in \Naturals. \forall n > n_0. \infinitepath(n) \in R_i\}] = 1.
	\end{equation*}
	For any BSCC $R \in \Bsccs(\MC)$ and states $s, s' \in R$, we have $\ProbabilityMC<\MC, s>[\reach \{s'\}] = 1$.
\end{lemma}
\paragraph*{Stationary Distribution} \label{sec:prelim:steadystate}
Given a state $\initialstate$, the \emph{stationary distribution} (also known as \emph{steady-state} or \emph{long-run distribution}) of a Markov chain intuitively describes, for each state $s$, the probability for the system to be at this particular state at an arbitrarily chosen step \enquote{at infinity}.
There are several ways to define this notion.
In particular, there is a subtle difference between the \emph{limiting} and \emph{stationary distribution}, which however coincide for \emph{aperiodic} MC.
For the sake of readability, we omit this distinction and assume w.l.o.g.\ that all MCs we deal with are aperiodic.
See \iftoggle{arxiv}{\cref{app:periodic}}{\cite[App.~A.1]{arxiv}} for further discussion.
Our definition follows the view of \cite[Def.~10.79]{DBLP:books/daglib/0020348}; see \cite[Sec.~A.4]{DBLP:books/wi/Puterman94} for a different approach.
\begin{figure}[t]
	\centering
	\begin{tikzpicture}[auto,initial text=,xscale=2]
		\node[state,initial above] at (0, 0) (s) {$s$};
		\node[state] at (-1, 0) (p) {$p$};
		\node[state] at (1, 0) (q1) {$q_1$};
		\node[state] at (2, 0) (q2) {$q_2$};

		\path[probedge]
			(s) edge[swap] node[prob] {$0.5$} (p)
			(s) edge node[prob] {$0.5$} (q1)

			(p) edge[loop above] node[prob] {$1.0$} (p)

			(q1) edge[loop above] node[prob] {$0.5$} (q1)
			(q1) edge[bend left=40] node[prob] {$0.5$} (q2)
			(q2) edge[loop above] node[prob] {$0.9$} (q2)
			(q2) edge[bend left=40] node[prob] {$0.1$} (q1)
		;
	\end{tikzpicture}
	\caption{
		Example MC to demonstrate the stationary distribution.
		We have that $\steadystate<\MC, s> = \{p \mapsto \frac{1}{2}, s \mapsto 0, q_1 \mapsto \frac{1}{2} \cdot \frac{1}{6}, q_2 \mapsto \frac{1}{2} \cdot \frac{5}{6}\}$.
	} \label{fig:steady_state_example}
\end{figure}
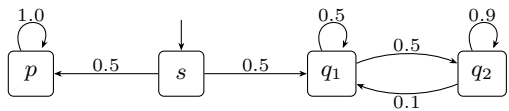
\begin{definition} \label{def:stationary}
	Fix a Markov chain $\MC = (\States, \mctransitions)$ and initial state $\initialstate$.
	Let $\pi_{\MC, \initialstate}^n(s) \coloneqq \ProbabilityMC<\MC, \initialstate>[\stepreach<n>\{s\}]$ the probability that the system is at state $s$ in step $n$.
	Then, $\steadystate<\MC, \initialstate>(s) \coloneqq \lim_{n \to \infty} \frac{1}{n} \sum_{i = 1}^n \pi_{\MC, \initialstate}^i(s)$ is the stationary distribution of $\MC$.
\end{definition}
See \cref{fig:steady_state_example} for an example.
Whenever the reference is clear from context, we omit the respective subscripts from $\steadystate<\MC, \initialstate>$.

We briefly recall the classical approach to compute stationary distributions (see e.g.\ \cite[Sec.~4.7]{kulkarni2016modeling}).
By \cref{stm:mc_bscc_properties}, almost all runs eventually end up in a BSCC.
Thus, $\steadystate(s) = 0$ for all states $s$ not in a BSCC, or, dually, $\sum_{s \in B} \steadystate(s) = 1$ for $B = \Union_{R \in \Bsccs(\MC)} R$.
Moreover, once in a BSCC, we always obtain the same stationary distribution, irrespective of through which state we entered the BSCC.
Formally, for each BSCC $R \in \Bsccs(\MC)$ and $s, s' \in R$, we have that $\steadystate<\MC, s> = \steadystate<\MC, s'> = \steadystate<\MCrestricted<R>, s>$, i.e.\ each BSCC $R$ has a unique stationary distribution, which we denote by $\steadystate<R>$.
Note that $\support(\steadystate<R>) = R$, i.e.\ $\steadystate<R>(s) \neq 0$ if and only if $s \in R$.
Together, we observe that the stationary distribution of a Markov chain decomposes into (i)~the steady state distribution in each BSCC and (ii)~the probability to end up in a particular BSCC.
More formally, for any state $s \in \States$
\begin{equation}
\steadystate<\MC, \initialstate>(s) = {\sum}_{R \in \Bsccs(\MC)} \ProbabilityMC<\MC, \initialstate>[\reach R] \cdot \steadystate<R>(s). \label{eq:steady_state_decomposition}
\end{equation}
%
%\begin{theorem} \label{stm:steady_state_decomposition}
%	For any $s \in \States$ we have
%	$
%		\steadystate<\MC, \initialstate>(s) = {\sum}_{R \in \Bsccs(\MC)} \ProbabilityMC<\MC, \initialstate>[\reach R] \cdot \steadystate<R>(s).
%	$
%\end{theorem}
%
Consider the example of \cref{fig:steady_state_example}:
We have two BSCCs, $\{p\}$ and $\{q_1, q_2\}$, which both are reached with probability $\frac{1}{2}$, respectively.
The overall distribution $\steadystate<\MC, s>$ then is obtained from $\steadystate<\{p\}> = \{p \mapsto 1\}$ and $\steadystate<\{q_1, q_2\}> = \{q_1 \mapsto \frac{1}{6}, q_2 \mapsto \frac{5}{6}\}$.

As mentioned, we can compute reachability probabilities in Markov chains by solving \cref{eq:reachability_value_fixpoint}.
Thus, the remaining concern is to compute $\steadystate<R>$, i.e.\ the stationary distribution of $\MCrestricted<R>$.
In this case, i.e.\ Markov chains comprising a single BSCC, the steady state distribution is the unique fixed point of the transition function (up to rescaling).
By defining the row transition matrix of $\MC$ as $P_{i, j} = \mctransitions(i, j)$, we can reformulate this property in terms of linear algebra.
In particular, we have that $P \cdot \steadystate<R> = \steadystate<R>$, or, in other words, $(P - I) \cdot \steadystate<R> = \vec{0}$, where $I$ is an appropriately sized identity matrix \cite[Thm.~A.2]{DBLP:books/wi/Puterman94}.
This equation again can be solved by classical methods from linear algebra.
In summary, we (i)~compute $\Bsccs(\MC)$, (ii)~for each BSCC $R$, compute $\steadystate<R>$ and $\ProbabilityMC<\MC, \initialstate>[\reach R]$, and (iii)~combine according to \cref{eq:steady_state_decomposition}.

However, as also mentioned in the introduction, precisely solving linear equation systems may not scale well, both due to time as well as memory constraints. % to models with billions of states, not only because of time but also space constraints.
Thus, we also are interested in relaxing the problem slightly and instead \emph{approximating} the stationary distribution up to a given precision of $\varepsilon > 0$.
\begin{framed}\textbf{Problem Statement}
	Given a Markov chain $\MC$ and precision requirement $\varepsilon > 0$, compute bounds $l, u : \States \to [0, 1]$ such that (i)~$\max_{s \in \States} u(s) - l(s) \leq \varepsilon$ and (ii)~for all $s \in \States$ we have $l(s) \leq \steadystate<\MC, \initialstate>(s) \leq u(s)$.
\end{framed}
\paragraph*{Approximate Solutions}
Aiming for approximations is not a new idea; to achieve practical performance, current model checkers employ approximate, iterative methods by default for most queries (typically a variant value iteration).
In particular, this also is the case for stationary distribution:
Instead of solving the equation system for each BSCC $R$ precisely, we can approximate the solution by, e.g., the \emph{power method}.
This essentially means to repeatedly apply the transition matrix (of the model restricted to the BSCC) to an initial vector $v_0$, i.e.\ iterating $v_{n+1} = P_R \cdot v_n$ (or $v_{n+1} = P_R^n \cdot v_1$).
Similarly, the reachability probability for each BSCC then also is approximated by value iteration.

%\footnote{
%	Note the subtle difference to other VI approaches:
%	Typically, values are \enquote{back-propagated} -- the value of $s$ is computed based on its successor values -- while here the values are \enquote{pushed forward}.}
It is known that (for aperiodic MC) $\lim_{n \to \infty} v_n = \steadystate<R>$ (see e.g.\ \cite{DBLP:books/daglib/0017293,DBLP:conf/qest/KatoenZ06,DBLP:journals/amc/NesterovN15}), however convergence up to a precision of $\varepsilon$ may take exponential time in the worst case.
Moreover, there is no known stopping criterion which allows us to detect that we have converged and stop the computation early.
Yet, similar to reachability \cite{DBLP:conf/atva/BrazdilCCFKKPU14,DBLP:journals/tcs/HaddadM18}, current model checkers employ this method without a sound stopping criterion, leading to potentially arbitrarily wrong results, as we show in our evaluation (\cref{fig:power_method_error}).
See \cite{DBLP:conf/qest/KatoenZ06} for a related, in-depth discussion of these issues in the context of CTMC.

We thus want to find efficient methods to derive safe bounds on the stationary distribution of a BSCC with a correct stopping criterion and combine it with correct reachability approximations to obtain an overall fast and sound approximation.
To this end, we exploit two further concepts.

\paragraph*{Partial Exploration}
Recent works \cite{DBLP:conf/atva/BrazdilCCFKKPU14,DBLP:conf/cav/AshokCDKM17,DBLP:journals/lmcs/KretinskyM20,pet} demonstrate the applicability of \emph{partial exploration} to a variety of problems associated with probabilistic systems such as reachability.
Essentially, the idea is to \enquote{omit} parts of the system which can be proven to be irrelevant for the result, instead focussing on important areas of the system.
Of course, by omitting parts of the system, we may incur a small error.
As such, these approaches naturally aim for approximate solutions.
%In our case, this introduces an interesting trade-off:
%We can either explore more to reduce the error of the unknown or obtain more precision by solving already known areas more precisely.
%

\paragraph*{Mean payoff} \label{sec:prelim:mean_payoff}
We make use of another property, namely \emph{mean payoff} (also known as \emph{long-run average reward}).
We provide a brief overview and direct to e.g.\ \cite[Chp.~8 \& 9]{DBLP:books/wi/Puterman94} or \cite{DBLP:conf/cav/AshokCDKM17} for more information.
Mean payoff is specified by a Markov chain and a \emph{reward function} $\reward : \States \to \Reals$, assigning a reward to each state.
Given an infinite path $\infinitepath = s_1 s_2 \cdots$, this naturally induces a stream of rewards $\reward(\infinitepath) \coloneqq \reward(s_1) \reward(s_2) \cdots$.
The mean payoff of this path then equals the average reward obtained in the limit, $\meanpayoff'_\reward(\infinitepath) \coloneqq \liminf_{n \to \infty} \frac{1}{n} \sum_{i = 1}^n \reward(s_i)$.
(The limit might not be defined for some paths, hence considering the $\liminf$ is necessary.)
Finally, the mean payoff of a state $s$ is the \emph{expected mean payoff} according to $\ProbabilityMC<\MC, s>$, i.e.\ $\meanpayoff_\reward(s) \coloneqq \ExpectationMC<\MC, s>[\meanpayoff'_\reward]$.

Classically, mean payoff is computed by solving a linear equation system \cite[Thm.~9.1.2]{DBLP:books/wi/Puterman94}.
Instead, we can also employ value iteration to approximate the mean payoff, however with a slight twist.
We iteratively compute the \emph{expected total reward}, i.e.\ the expected sum of rewards obtained after $n$ steps, by iterating $v_{n+1}(s) = \reward(s) + \ExpectedSumMC{\mctransitions}{s}{v_n}$.
It turns out that the \emph{increase} $\Delta_n(s) = v_{n+1}(s) - v_n(s)$ approximates the mean payoff, i.e.\ $\meanpayoff_\reward(s) = \lim_{n \to \infty} \Delta_n(s)$ \cite[Thm.~9.4.5 a)]{DBLP:books/wi/Puterman94}.
Moreover, we have $\min_{s' \in S} \Delta_n(s') \leq \meanpayoff_\reward(s) \leq \max_{s' \in S} \Delta_n(s')$, yielding a correct stopping criterion \cite[Thm.~9.4.5 b)]{DBLP:books/wi/Puterman94}.
Finally, on BSCCs these upper and lower bounds always converge \cite[Cor.~9.4.6 b)]{DBLP:books/wi/Puterman94}, yielding termination guarantees.
We provide further details on VI for mean payoff in \iftoggle{arxiv}{\cref{app:mean_payoff}}{\cite[App.~A.3]{arxiv}}.
\section{Building Blocks} \label{sec:blocks}

To arrive at a practical algorithm approximating the stationary distribution, we propose to employ sampling-based techniques, inspired by, e.g.\ \cite{DBLP:conf/atva/BrazdilCCFKKPU14,DBLP:conf/cav/AshokCDKM17,DBLP:journals/lmcs/KretinskyM20}.
Intuitively, these approaches repeatedly sample paths and compute bounds on a single property such as reachability or mean payoff.
The sampling is designed to follow probable paths with high probability, hence the computation automatically focuses on the most relevant parts of the system.
Additionally, by building the system \emph{on the fly}, construction of hardly reachable parts of the system may be avoided altogether, yielding immense speed-ups for some models (see, e.g., \cite{DBLP:journals/lmcs/KretinskyM20} for additional background).
We apply a series of tweaks to the original idea to tailor this approach to our use case, i.e.\ approximating the stationary distribution.

In this section, we present the \enquote{building blocks} for our approximate approach.
In the spirit of \cref{eq:steady_state_decomposition}, we discuss how we handle a single BSCC and how to approximate the reachability probabilities of all BSCCs.
In the following section, we then combine these two approaches in a non-trivial manner.
\subsection{Bounds in BSSCs through Mean Payoff} \label{sec:blocks:bscc_mean_payoff}
It is well known that the mean payoff can be computed directly from the stationary distribution \cite[Prop.~8.1.1]{DBLP:books/wi/Puterman94}, namely:
\begin{equation} \label{eq:mean_payoff_through_steady_state}
	\meanpayoff_\reward(s) = {\sum}_{s' \in \States} \steadystate<\MC, s>(s') \cdot \reward(s')
\end{equation}
In this section, we propose the opposite, namely computing the stationary distribution of a BSCC through mean payoff queries.
Fix a Markov chain $\MC = (\States, \mctransitions)$ which comprises a single BSCC, i.e.\ $\States \in \Bsccs(\MC)$, and define $\reward(s') = \indicator{\{s\}}(s')$, i.e.\ $1$ for $s$ and $0$ otherwise.
Then, the mean payoff corresponds to the frequency of $s$ appearing, i.e.\ the stationary distribution.
Formally, we have that $\steadystate<\MC, \initialstate>(s) = \meanpayoff_\reward(s')$ for any state $s'$ (in a BSCC, all states have the same value).
This also follows directly by inserting in \cref{eq:mean_payoff_through_steady_state}.
So, naively, for each state of the BSCC, we can solve a mean payoff query, and from these results obtain the overall stationary distribution.

At first, this may seem excessive, especially considering that computing the complete stationary distribution is as hard as determining the mean payoff for one state (both can be obtained by solving a linearly sized equation system).
%Both are PTIME and the solution of linear sized equation systems.
However, this idea yields some interesting benefits.
Firstly, using the approximation approach discussed in \cref{sec:prelim:mean_payoff}, we obtain a practical approximation scheme with converging bounds for each state.
As such, we can quickly stop the computation if the bounds converge fast.
Moreover, we can pause and restart the computation for each state, which we will use later on in order to focus on crucial states.
Finally, observe that $\steadystate<R>$ is a distribution.
Thus, having lower bounds on some states actually already yields upper bounds for remaining states.
Formally, for some lower bound $l : \States \to [0, 1]$, we have $\steadystate<R>(s) \leq 1 - \sum_{s' \in \States, s' \neq s} l(s')$.
If during our computation it turns out that a few states are actually visited very frequently, i.e.\ the sum of their lower bounds is close to $1$, we can already stop the computation without ever investigating the other states.
Note that this only is possible since we obtain provably correct bounds.

\begin{algorithm}[t]
	\caption{Approximate Stationary Distribution in BSCC} \label{alg:approx_bscc}
	\begin{algorithmic}[1]
		\Require Markov chain $\MC = (\States, \mctransitions)$ with $\Bsccs(\MC) = \{\States\}$
		\Ensure Bounds $l, u$ on stationary distribution $\steadystate<S>$.
		\State $n \gets 1$
		\For {$s \in \States$} $l_1(s) \gets 0$, $u_1(s) \gets 1$
		\EndFor
		\For {$s \in \States$}
			\State $m \gets 1$, $v_1 \gets \Call{InitGuess}{s}$
			
			\While{not $\Call{ShouldStop}{s, m, \Delta_m}$} \Comment{Iterate until some stopping criterion}
				\For {$s' \in \States$} $v_{m+1}(s') \gets \indicator{\{s\}}(s') + \ExpectedSumMC{\mctransitions}{s'}{v_m}$ \Comment{Mean payoff VI for $s$}
				\EndFor
				\State $m \gets m + 1$
			\EndWhile

			\State $l'_n(s) \gets \max\big(l_n(s), \min_{s' \in \States} \Delta_m(s')\big)$, $u'_n(s) \gets \min\big(u_n(s), \max_{s' \in \States} \Delta_m(s')\big)$
			\For {$s' \in \States \setminus \{s\}$} $l'_n(s') \gets l_n(s')$, $u'_n(s') \gets u_n(s')$
			\EndFor

			\For {$s' \in \States$} \Comment{Update bounds based on current results (optional)}
				\State $l_{n + 1}(s') \gets \max\big( l'_n(s'), 1 - \sum_{s'' \in \States, s'' \neq s'} u'_n(s'') \big)$
				\State $u_{n + 1}(s') \gets \min\big( u'_n(s'), 1 - \sum_{s'' \in \States, s'' \neq s'} l'_n(s'') \big)$
			\EndFor
			\State $n \gets n + 1$ and copy all unchanged values from $n$ to $n + 1$
		\EndFor
		\State \Return $(l_n, u_n)$
	\end{algorithmic}
\end{algorithm}

Combining these ideas, we present our first algorithm template in \cref{alg:approx_bscc}.
We solve each state separately, by applying the classical value iteration approach for mean payoff until a termination criterion is satisfied.
To allow for modifications, we leave the definition of several sub-procedures open.
Firstly, $\Call{InitGuess}{}$ initializes the value vector for each mean payoff computation.
We can naively choose $0$ everywhere, obtain an initial guess by heuristics, or re-use previously computed values.
Secondly, $\Call{ShouldStop}{}$ decides when to stop the iteration for each state.
A simple choice is to iterate until $\max \Delta_m(s) - \min \Delta_m(s) < \varepsilon$ for some precision requirement $\varepsilon$.
By results on mean payoff, we can conclude that in this case the stationary distribution is computed with a precision of $\varepsilon$.
However, as we argue later on, more sophisticated choices are possible.
Finally, the order in which states are chosen is not fixed.
Indeed, any order yields correct results, however heuristically re-ordering the states may also bring practical benefits.

Before we continue, we briefly argue that the algorithm is correct.
\begin{theorem}
	The result returned by \cref{alg:approx_bscc} is correct for any MC $\MC = (\States, \mctransitions)$ with $\Bsccs(\MC) = \{\States\}$.
\end{theorem}
\begin{proof}[Sketch]
	Correctness of the mean payoff iteration follows from the definition of the reward function, \cref{eq:mean_payoff_through_steady_state}, and the correctness of value iteration for mean payoff \cite[Sec.~8.5]{DBLP:books/wi/Puterman94}.
	In particular, note that the states of the MC form a single BSCC and the model is \emph{unichain} (see \cite[Chp.~A]{DBLP:books/wi/Puterman94}), implying that all states have the same value.
	For $l$ and $u$, we prove correctness inductively.
	The initial values are trivially correct.
	The updates based on the mean payoff computation are correct by the above arguments and by induction hypothesis:
	The maximum of two correct lower bounds still is a lower bound, analogous for the upper bound.
	The updates based on the bounds are correct since $\steadystate<R>$ is a distribution and $l'$, $u'$ are correct bounds. \qed
\end{proof}
We deliberately omit introducing an explicit precision requirement in the algorithm, since we will use it as a building block later on.
%Convergence to a given precision of $\varepsilon$ may require exponential time in the worst case.
%However, this approach typically is fast in practice, as shown in our experimental evaluation.
\begin{remark} \label{rem:space}
	A variant of this approach also allows for memory savings:
	By handling one state at a time, we only need to store linearly many additional values (in the number of states) at any time, while an explicit equation system may require quadratic space.
	This only yields a constant factor improvement if the system is represented explicitly (storing $\mctransitions$ requires as much space), however can be of significant merit for symbolically encoded systems.
	Note that this comes at a cost:
	As we cannot stop and resume the computation for different states, we have to determine the correct result up to the required precision immediately.
\end{remark}

\subsection{Reachability and Guided Sampling}
As mentioned before, the second challenge to obtain a stationary distribution is the reachability probability for each BSCC.
We employ a sampling-based approach using insights from \cite{DBLP:conf/atva/BrazdilCCFKKPU14}.
There, the authors considered a single reachability objective, i.e.\ a single value per state.
In contrast, we need to bound reachability probabilities for each BSCC.
For now, suppose that all BSCCs are already discovered and their respective stationary distribution is already computed (or approximated).
In other words, we have for each BSCC $R \in \Bsccs(\MC)$ bounds $l^R, u^R : R \to [0, 1]$ with $l_R(s) \leq \steadystate<R>(s) \leq u_R(s)$, and we want to obtain bounds on the stationary distribution, i.e.\ functions $l$, $u$ such that $l(s) \leq \steadystate<\MC, \initialstate>(s) \leq u(s)$.
We propose to additionally compute bounds on the probability to reach each BSCC $R$, i.e.\ functions $l^{\reach R}$ and $u^{\reach R}$ such that $l^{\reach R}(s) \leq \ProbabilityMC<\MC, s>[\reach R] \leq u^{\reach R}(s)$.
By \cref{eq:steady_state_decomposition}, we then have for each state $s$ a bound on the stationary distribution
\begin{equation*}
	{\sum}_{R \in \Bsccs(\MC)} l^{\reach R}(\initialstate) \cdot l^R(s) \leq \steadystate<\MC, \initialstate>(s) \leq {\sum}_{R \in \Bsccs(\MC)} u^{\reach R}(\initialstate) \cdot u^R(s).
\end{equation*}
%, i.e.\ an overall bound on the stationary distribution.

We take a route similar to \cite{DBLP:conf/atva/BrazdilCCFKKPU14}.
There, the algorithm essentially samples a path through the system, possibly guided by a heuristic, terminates the sampling based on several criteria, and then propagates the reachability value backwards along the path, repeating until termination.
We propose a simple modification, namely to sample until a BSCC is reached, and then propagate the reachability values of that particular BSCC back along the path.
Moreover, we can employ a similar trick as above:
Due to \cref{stm:mc_bscc_properties}, the reachability probabilities of BSCCs sum up to one, i.e.\ $\sum_{R \in \Bsccs(\MC)} \ProbabilityMC<\MC, s>[\reach R] = 1$ for every state $s$.
Hence, the sum of lower bounds also yields upper bounds for other BSCCs, even those we have never encountered so far.

\begin{algorithm}[t]
	\caption{Approximate BSCC Reachability} \label{alg:approx_reach}
	\begin{algorithmic}[1]
		\Require Markov chain $\MC = (\States, \mctransitions)$
		\Ensure For each BSCC $R$ bounds $l^{\reach R}, u^{\reach R}$ on the probability to reach $R$.
		\State $B \gets \Union_{R \in \Bsccs(\MC)} R$, $n \gets 1$
		\For {$R \in \Bsccs(\MC)$}
			\For {$s \in R$} $l^{\reach R}_1(s) \gets 1$, $u^{\reach R}_1(s) \gets 1$
			\EndFor
			\For {$s \in B \setminus R$} $l^{\reach R}_1(s) \gets 0$, $u^{\reach R}_1(s) \gets 0$
			\EndFor
			\For {$s \in \States \setminus B$} $l^{\reach R}_1(s) \gets 0$, $u^{\reach R}_1(s) \gets 1$
			\EndFor
		\EndFor
		\While {$\Call{ShouldSample}{}$} \Comment{Sample until some stopping criterion}
			\State $P \gets \Call{SampleStates}{}$ \Comment{Select states to update (e.g.\ sample a path)}
			\For {$R \in \Call{SelectUpdate}{P}$} \Comment{Select BSCCs to update}
				\For {$s \in P$}
					\State $l^{\reach R}_{n+1}(s) \gets \ExpectedSumMC{\mctransitions}{s}{l^{\reach R}_n}$
					\State $u^{\reach R}_{n+1}(s) \gets \ExpectedSumMC{\mctransitions}{s}{u^{\reach R}_n}$
				\EndFor
			\EndFor

			\For {$s \in \States$} \Comment{Update bounds based on current results (optional)}
				\For {$R \in \Bsccs(\MC)$}
					\State $l^{\reach R}_{n+1}(s) \gets \max\big(l^{\reach R}_n(s), 1 - \sum_{R' \in \Bsccs(\MC), R' \neq R} u^{R'}_n(s)\big)$
					\State $u^{\reach R}_{n+1}(s) \gets \min\big(u^{\reach R}_n(s), 1 - \sum_{R' \in \Bsccs(\MC), R' \neq R} l^{R'}_n(s)\big)$
				\EndFor
			\EndFor
			\State $n \gets n + 1$ and copy unchanged values from $l^{\reach R}_n$ and $u^{\reach R}_n$ to $l^{\reach R}_{n+1}$ and $u^{\reach R}_{n+1}$
		\EndWhile
		\State \Return $\{(l^{\reach R}, u^{\reach R}) \mid R \in \Bsccs(R)\}$
	\end{algorithmic}
\end{algorithm}

Our ideas are summarized in \cref{alg:approx_reach}.
As before, the algorithm leaves several choices open.
Instead of requiring to sample a path, our algorithm allows to select an arbitrary set of states to update.
We note that the exact choice of this sampling mechanism does not improve the worst case runtime.
However, as first observed in \cite{DBLP:conf/atva/BrazdilCCFKKPU14}, specially crafted \emph{guidance heuristics} can achieve dramatic practical speed-ups on several models.
Later on, we combine our two algorithms and derive such a heuristic.
For now, we briefly prove correctness.
\begin{theorem}
	The result returned by \cref{alg:approx_reach} is correct for any MC $\MC = (\States, \mctransitions)$ with $\Bsccs(\MC) = \{\States\}$.
\end{theorem}
\begin{proof}[Sketch]
	Similar to the previous algorithm, we prove correctness by induction.
	The initial values for $l^{\reach R}$ and $u^{\reach R}$ are correct.
	Then, assume that $l^{\reach R}_n$ and $u^{\reach R}_n$ are correct bounds.
	The correctness of the back propagation updates follows directly by inserting in \cref{eq:reachability_value_fixpoint} (or other works on interval value iteration \cite{DBLP:conf/atva/BrazdilCCFKKPU14,DBLP:journals/tcs/HaddadM18}).
	Updates based on the bounds in other states are correct by \cref{stm:mc_bscc_properties} -- the sum of all BSCC reachability probabilities is 1.
	Together, this yields correctness of the bounds computed by the algorithm. \qed
\end{proof}
To obtain termination, it is sufficient to require that every state eventually is selected \enquote{arbitrarily often} by \Call{SampleStates}{}.
However, as before, we delegate the termination proof to our combined algorithm in the following section.
\section{Dynamic Computation with Partial Exploration} \label{sec:framework}
Recall that our overarching goal is to approximate the stationary distribution through \cref{eq:mean_payoff_through_steady_state}.
In the previous section, we have seen how we can (i)~obtain approximations for a given BSCC and (ii)~how to approximate the reachability probabilities of all BSCCs through sampling.
However, the naive combination of these algorithms would require us to compute the set of all BSCCs, approximate the stationary distribution in each of them until a fixed precision, and additionally approximate reachability for each of them.

We now combine both ideas to obtain a sampling-based algorithm, capable of partial exploration, that focusses computation on relevant parts of the system.
In particular, we construct the system dynamically, identify BSCCs on the fly, and interleave the exploration with both the approximation inside each explored BSCC (\cref{alg:approx_bscc}) and the overall reachability computation (\cref{alg:approx_reach}).
Moreover, we focus computation on BSCCs which are likely to be reached and thus have a higher impact on the overall error of the result.
Together, our approach roughly performs the following steps until the required precision is achieved:
\begin{itemize}
	\item Sample a path through the system, guided by a heuristic,
	\item check if a new BSCCs is discovered or sampling ended in a known BSCC,
	\item refine bounds on the stationary distribution in the reached BSCC, and
	\item propagate reachability bounds and additional information along the path.
\end{itemize}
We first formalize a generic framework which can instantiate the classical, precise approach as well as our approximation building blocks and then explain our concrete variant of this framework to efficiently obtain $\varepsilon$-precise bounds.
\subsection{The Framework}
Since our goal is to allow for both precise as well as approximate solutions, we phrase the framework using lower and upper bounds together with abstract refinement procedures.
We first explain our algorithm and how it generalizes the classical approach.
Then, we prove its correctness under general assumptions.
Finally, we discuss several approximate variants.

\begin{algorithm}[t]
	\caption{Stationary Distribution Computation Framework} \label{alg:framework}
	\begin{algorithmic}[1]
		\Require Markov chain $\MC = (\States, \mctransitions)$, initial state $\initialstate$, precision $\varepsilon > 0$
		\Ensure $\varepsilon$-precise bounds $l, u$ on the stationary distribution $\steadystate<\MC, \initialstate>$

		\For {$s \in \States$} \Comment{Initial bounds for all possible BSSCs that can be discovered}
			\State $l^{\reach \circ}_1(s) = 0$, $u^{\reach \circ}_1(s) = 1$, $l^\circ_1(s) \gets 0$, $u^\circ_1(s) \gets 1$ \label{line:framework:init}
		\EndFor
		\State $n \gets 1$, $\mathcal{B}_1 \gets \emptyset$

		\While {$\big( 1 - \sum_{R \in \mathcal{B}_n} l^{\reach R}_n(\initialstate) \big) + \sum_{R \in \mathcal{B}_n} \big( l^{\reach R}_n(\initialstate) \cdot \max_{s \in \States} (u^R_n(s) - l^R_n(s)) \big) > \varepsilon$} \label{line:framework:loop}
			\State $n \gets n + 1$

			\State $\mathcal{B}_n \gets \Call{UpdateBSSCs}{}$, $B_n \gets \Union_{R \in \mathcal{B}_n} R$ \Comment{Discover new BSCCs} \label{line:framework:bscc_explore}
			\For {$R \in \mathcal{B}_n \setminus \mathcal{B}_{n-1}$, $s \in R$} \Comment{Update trivial reach bounds}
				\State $l^{\reach R}_n(s) \gets 1$ \label{line:framework:bscc_explore_reach_lower_bounds} \Comment{$s \in R$ surely reaches $R$}
				\For {$\circ \neq R$} $u^{\reach \circ}_n(s) \gets 0$ \label{line:framework:bscc_explore_reach_upper_bounds} \Comment{$s \in R$ reaches no other BSCC}
				\EndFor
			\EndFor
			\For {$R \in \Call{SelectDistributionUpdates}{\mathcal{B}_n} \intersection \mathcal{B}_n$}
				\State $(l^R_n, u^R_n) \gets \Call{RefineDistribution}{R}$ \label{line:framework:bssc_update_bounds} \Comment{Update BSCC bounds}
			\EndFor
			\For {$R \in \Call{SelectReachUpdates}{\mathcal{B}_n} \intersection \mathcal{B}_n$}
				\State $(l^{\reach R}_n, u^{\reach R}_n) \gets \Call{RefineReach}{R}$ \label{line:framework:reach_update_bounds} \Comment{Update reachability bounds}
			\EndFor

			\State Copy unchanged variables from $n - 1$ to $n$ \label{line:framework:copy}
		\EndWhile
		\State $L \gets \sum_{R \in \mathcal{B}_n} l^{\reach R}_n(\initialstate)$ \label{line:framework:result_lower_bounds_sum}
		\For {$R \in \mathcal{B}_n$, $s \in R$}
			\State $l(s) \gets l^{\reach R}_n(\initialstate) \cdot l^R_n(s)$
			\State $u(s) \gets \min(u^{\reach R}_n(\initialstate), 1 - L + l^{\reach R}_n(\initialstate)) \cdot u^R_n(s)$ \label{line:framework:compute_result}
		\EndFor
		\For {$s \in \States \setminus B_n$} $l(s) \gets 0$, $u(s) \gets 0$ \label{line:framework:result_zero_states}
		\EndFor
		\State \Return $(l, u)$
	\end{algorithmic}
\end{algorithm}

\Cref{alg:framework} essentially repeats three steps until the termination condition in \cref{line:framework:loop} is satisfied.
First, we update the set of known BSCCs through \Call{UpdateBSSCs}{}.
In the classical solution, this function simply computes $\Bsccs(\MC)$ once; our on-the-fly construction would repeatedly check for newly discovered BSCCs, dynamically growing the set $\mathcal{B}_n$.
Then, we select BSCCs for which we should update the stationary distribution bounds.
The classical solution solves the fixed point equation we have discussed in \cref{sec:prelim:steadystate} for all BSCCs, i.e.\ \Call{SelectDistributionUpdates}{} yields $\Bsccs(\MC)$ and \Call{RefineDistribution}{} the precisely computed values both as upper and lower bounds.
Alternatively, we could, for example, select a single BSCC and apply a few iterations of \cref{alg:approx_bscc}.
Next, we update reachability bounds for a selected set of BSCCs.
Again, the classical solution solves the reachability problem precisely for each BSCC through \cref{eq:reachability_value_fixpoint}.
Instead, we could employ value iteration as suggested by \cref{alg:approx_reach}.

Before we present our variant, we prove correctness under weak assumptions.
We note a subtlety of the termination condition:
One may assume that upper bounds on the reachability are required to bound the overall error caused by each BSCC.
Yet, as we show in the following theorem, \emph{lower} bounds are sufficient.
The upper bound is implicitly handled by the first part of the termination condition.
\begin{theorem} \label{stm:framework_correct}
	The result returned by \cref{alg:framework} is correct, i.e.\ $\varepsilon$ precise bounds on the stationary distribution, if (i)~$\mathcal{B}_n \subseteq \mathcal{B}_{n+1} \subseteq \Bsccs(\MC)$ for all $n$, and (ii)~\Call{RefineDistribution}{} and \Call{RefineReach}{} yield correct, monotone bounds.
\end{theorem}
The proof can be found in \iftoggle{arxiv}{\cref{app:proof:framework_correct}}{\cite[App.~B.1]{arxiv}}.
\begin{remark}
	Technically, the algorithm does not need to track explicit upper bounds on the reachability of each BSCC at all.
	Indeed, for a BSCC $R \in \mathcal{B}_n$, we could use $1 - \sum_{R' \in \Bsccs(\MC) \setminus \{R\}} l^{\reach R'}_n(s)$ as upper bound and still obtain a correct algorithm.
	However, tracking a separate upper bound is easier to understand and has some practical benefits for the implementation.
\end{remark}
We exclude a proof of termination, since this strongly depends on the interplay between the functions left open.
We provide a general, technical criterion together with a proof in \iftoggle{arxiv}{\cref{app:proof:framework_termination}}{\cite[App.~B.2]{arxiv}}.
Intuitively, as one might expect, we require that eventually \Call{UpdateBSSCs}{} identifies all relevant BSCCs, \Call{SelectDistributionUpdates}{} and \Call{SelectReachUpdates}{} select all relevant BSCCs, and \Call{RefineDistribution}{} and \Call{RefineReach}{} converge to the respective true value.
In the following, we present a concrete template which satisfies this criterion.
%
%\begin{remark}
%	We note that our formulation \enquote{The algorithm terminates if, assuming that it does not terminate, ...} is slightly imprecise.
%	Instead we could consider a formulation like \enquote{if for every execution there exists a step $N$ such that ...}.
%	This would make the statement more formal, but significantly more cumbersome to read.
%	Hence, for the sake of simplicity, we opted for our formulation.
%	Alternatively, we could also replace the loop condition in \cref{line:framework:loop} with \enquote{true}, i.e.\ loop forever, and show that the condition would eventually be satisfied, similar to \cite{DBLP:conf/atva/BrazdilCCFKKPU14}.
%\end{remark}
%
%
%
\subsection{Sampling-Based Computation} \label{sec:framework:sampling}
We present our instantiation of \cref{alg:framework} using guided sampling and heuristics.
Since the details of the sampling guidance heuristic are rather technical, we focus on how the template functions \Call{UpdateBSSCs}{}, \Call{SelectDistributionUpdates}{}, \Call{RefineDistribution}{}, \Call{SelectReachUpdates}{}, and \Call{RefineReach}{} are instantiated. % then discuss the challenges of sampling in more detail, and finally present our proposed instantiation.
For now, the reader may assume that states are, e.g., selected by sampling random paths through the system.
\begin{itemize}
	\item \Call{UpdateBSSCs}{}:
		We track the set of \emph{explored} states, i.e.\ states which have already been sampled at least once.
		On these, we search for BSCCs whenever we repeatedly stop sampling due to a state re-appearing.

	\item \Call{SelectDistributionUpdates}{}:
		If we stopped sampling due to entering a known BSCC, we update the bounds of this single one, otherwise none.

	\item \Call{RefineDistribution}{}:
		We employ \cref{alg:approx_bscc} to refine the bounds until the error over all states is halved.
		%If we refine distributions for a BSCC for the first time, we run a few simulations to obtain an initial guess for $v_1$.

	\item \Call{SelectReachUpdates}{}:
		We refine the reach values for all sampled states.

	\item \Call{RefineReach}{}:
		If we stopped sampling due to entering a BSCC, we back-propagate the reachability bounds for this BSCC in the spirit of \cref{alg:approx_reach}, i.e.\ for all sampled states set $l^{\reach R}_{n+1}(s) = \ExpectedSumMC{\mctransitions}{s}{l^{\reach R}_n}$ and $u^{\reach R}_{n+1}(s) = \ExpectedSumMC{\mctransitions}{s}{u^{\reach R}_n}$.
\end{itemize}
\newcommand{\errorbound}{\mathsf{err}} %
We prove that this yields correct results and terminates with probability 1 through \cref{stm:framework_correct}.
Note that this description leaves exact details of the sampling open.
Thus, we prove termination using (weak) conditions on the sampling mechanism.
For readability, we define the shorthand $\errorbound^R_n = \max_{s \in R} u_n^R(s) - l_n^R(s)$ denoting the overall error of the stationary distribution in BSCC $R$ and $\errorbound^{\reach R}_n(s) = u^{\reach R}_n(s) - l^{\reach R}_n(s)$ the error bound on the reachability of $R$ from $s$.
\begin{theorem} \label{stm:sampling}
	\Cref{alg:framework} instantiated with our sampling-based approach yields correct results and terminates with probability 1 if, with probability 1,
	\begin{enumerate}[label=(S.\roman*),labelsep=*,leftmargin=*]
		\item \label{stm:sampling:core}
		the sampled states $P \subseteq \States$ satisfy $\ProbabilityMC<\MC, \initialstate>[\reach \setcomplement{P}] < \frac{\varepsilon}{4}$ ($P$ is a \emph{$\frac{\varepsilon}{4}$-core} \cite{DBLP:journals/lmcs/KretinskyM20}),
		\item \label{stm:sampling:initial}
		the initial state is sampled arbitrarily often, and
		\item \label{stm:sampling:successor}
		for each state $s$ sampled arbitrarily often, every successor $s' \in P$ with $E_n(s') \coloneqq \max_{R \in \mathcal{B}_n} u^{\reach R}_n(s') \cdot \errorbound^R_n + \max_{R \in \mathcal{B}_n} \errorbound^{\reach R}_n(s) \geq \frac{\varepsilon}{4 (\cardinality{\mathcal{B}_n} + 1)}$ is sampled arbitrarily often,
	\end{enumerate}
	where \enquote{arbitrarily often} means that if the algorithm would not terminate, this would happen infinitely often.
\end{theorem}
The proof can be found in \iftoggle{arxiv}{\cref{app:proof:sampling}}{\cite[App.~B.3]{arxiv}}.

Due to space constraints, we omit an in-depth description of our sampling method and only provide a brief summary here.
In summary, our algorithm first selects a \enquote{sampling target} which is either \enquote{the unknown}, i.e.\ states not seen so far, to encourage exploration in the style of \cite{DBLP:journals/lmcs/KretinskyM20}, or a known BSCC, to bias sampling towards it.
We select a choice randomly, weighted by its current potential influence on the precision.
The sampling process is guided by the chosen target, taking actions which lead to the respective target with high probability.
In technical terms, we sample successors weighted by the upper bound on reachability probability times the transition probability.
Once the target is reached, we either explore the unknown, or improve precision in the reached BSCC.
Finally, information is back-propagated along the path.
Further details, in particular pitfalls we encountered during the design process, together with a complete instantiation of our algorithm can be found in \iftoggle{arxiv}{\cref{app:algorithm}}{\cite[App.~C]{arxiv}}.
\section{Experimental Evaluation} \label{sec:evaluation}

In this section, we evaluate our approaches, comparing to both our own reference implementation using classical methods, as well as the established model checker PRISM \cite{DBLP:conf/cav/KwiatkowskaNP11}.
(The other popular model checkers Storm \cite{DBLP:conf/cav/DehnertJK017} and IscasMC/ePMC \cite{DBLP:conf/fm/HahnLSTZ14} do not directly support computing stationary distributions.)
We implemented our methods in Java based on PET \cite{pet}, running on consumer hardware (AMD Ryzen 5 3600).
To solve arising linear equation systems, we use the library \texttt{Jeigen v1.2}.
All executions are performed in a Docker container, restricted to a single CPU core and 8GB of RAM.
For our approximation approaches, we require a precision of $\varepsilon = 10^{-4}$.

\paragraph{Tools}
Aside from \texttt{PRISM}\footnote{
	We observed that the default hybrid engine typically is significantly slower than the \enquote{explicit} variant and thus use that one, see \iftoggle{arxiv}{\cref{app:results}}{\cite[App.~D]{arxiv}}.
}, we consider three variants of \cref{alg:framework}, namely
	\texttt{Classic}, the classical approach, solving each BSCC through a linear equation system and then approximating the reachability through PRISM (using interval iteration),
	\texttt{Naive}, the naive sampling approach, following the transition dynamics, and
	\texttt{Sample}, our sampling approach, selecting a target and steering towards it.
The sourcecode of our implementation used to run these experiments as well as all models and our data is available at \cite{zenodo}.
Moreover, the current version can be found at GitHub \cite{github}.

We mention two points relevant for the comparison.
First, as we show in the following, \texttt{PRISM} may yield wrong results due to a (too) simple computation.
As such, we should not expect that our correct methods are on par or even faster.
Second, our implementation employs conservative procedures to further increase quality of the result, such as compensated summation to mitigate numerical error due to floating-point imprecision, noticeably increasing computational effort.

\paragraph{Models}
We consider the PRISM benchmark suite\footnote{Obtained from \url{https://github.com/prismmodelchecker/prism-benchmarks}.} \cite{DBLP:conf/qest/KwiatkowskaNP12}, comprising several probabilistic models, in particular DTMC, CTMC, and MDP.
Since there are not too many Markov chains in this set, we obtain further models as follows.
For each CTMC, we consider the \emph{uniformized CTMC} (which preserves the steady state distribution), and for MDP we choose actions uniformly at random.
Unfortunately, \emph{all} models obtained this way either comprise only single-state BSCCs or the whole model is a single BSCC.
In the former case, our approximation within the BSCC is not used at all, in the latter, a sampling based approach needs to invest additional time to discover the whole system.
In order to better compare the performance of our mean payoff based approximation approach, in these cases we pre-explore the whole system and compute the stationary distribution directly through \cref{alg:approx_bscc}.
To compare the combined performance, we additionally consider a handcrafted model, named \textbf{branch}, which comprises both transient states as well as several non-trivial BSCCs.

We present selected results, highlighting different strengths and weaknesses of each approach.
An evaluation of the complete suite can be found in \iftoggle{arxiv}{\cref{app:results}}{\cite[App.~D]{arxiv}}.

%We consider the models \textbf{brp} (sending a file over a faulty connection with retransmissions), \textbf{crowds} (anonymity protocol for web browsing), and \textbf{nand} (NAND multiplexing from faulty devices).

%\paragraph{Experiments}
%For each model, we are interested in the overall runtime of each tool.
%Aside from time constraints, partial exploration in particular can yield memory advantages, since we only build a fraction of the model.
%Unfortunately, comparing these values directly is tricky due to several practical reasons.
%Firstly, our approach represents transitions explicitly, compared to a compact representation by MTBDDs.
%Secondly, Java allocates and manages memory on its own, yielding wrong results for metrics such as process peak memory.
%As such, we compare the number of explored states by each method as a proxy for overall memory requirement.

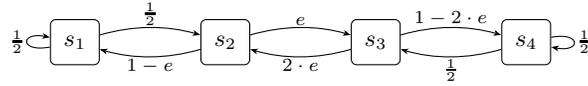
\begin{figure}[t]
	\centering
	\begin{tikzpicture}[auto,xscale=2]
		\node[state] at (0, 0) (s11) {$s_{2}$};
		\node[state] at (-1, 0) (s12) {$s_{1}$};
		\node[state] at (1, 0) (s21) {$s_{3}$};
		\node[state] at (2, 0) (s22) {$s_{4}$};

		\path[probedge]
			(s11) edge[bend left] node[prob] {$1 - e$} (s12)
			(s11) edge[bend left] node[prob] {$e$} (s21)

			(s21) edge[bend left] node[prob] {$1 - 2 \cdot e$} (s22)
			(s21) edge[bend left] node[prob] {$2 \cdot e$} (s11)

			(s12) edge[loop left,looseness=3] node[prob] {$\frac{1}{2}$} (s12)
			(s12) edge[bend left] node[prob] {$\frac{1}{2}$} (s11)

			(s22) edge[loop right,looseness=3] node[prob] {$\frac{1}{2}$} (s22)
			(s22) edge[bend left] node[prob] {$\frac{1}{2}$} (s21)
		;
	\end{tikzpicture}
	\caption{
		A small MC where PRISM reports wrong results for $e \leq 10^{-7}$.
		% By replacing $2 \cdot e$ with e.g.\ $10 \cdot e$, the magnitude of PRISM's error can be increased.
	} \label{fig:power_method_error}
\end{figure}

\paragraph{Correctness}

We discovered that \texttt{PRISM} potentially yields wrong results, due to an unsafe stopping criterion.
In particular, \texttt{PRISM} iterates the power method until the absolute difference between subsequent iterates is small, exactly as with its \enquote{unsafe} value iteration for reachability, as reported by e.g.\ \cite{DBLP:conf/atva/BrazdilCCFKKPU14}.
On the model from \cref{fig:power_method_error}, PRISM (with explicit engine) immediately terminates, printing a result of $\approx (\frac{1}{6}, \frac{1}{6}, \frac{1}{3}, \frac{1}{3})$.
However, the correct stationary distribution is $\approx (\frac{1}{9}, \frac{2}{9}, \frac{4}{9}, \frac{2}{9})$ (from left to right), which both of our methods correctly identify.
This behaviour is due to the small difference between first and second eigenvalue of the transition matrix, which in turn implies that the iterates of the power method only change by a small amount.
We note that on this example, PRISM's default hybrid engine eventually yields the correct result (after $\approx 10^8$ iterations) due to the used iteration scheme.
On small variation of the model (included in the artefact) it also terminates immediately with the wrong result.
%We did not observe such an error in the practical models we evaluated.

\begin{table}[t]
	\caption{
		Overview of our results.
		For each model, we list its parameters, overall size, and number of BSCCs, followed by the total execution time in seconds for each tool, TO denotes a timeout (300 seconds), MO a memout, and \texttt{err} an internal error.
		On systems comprising a single BSCC, the \texttt{Naive} and \texttt{Sample} approach coincide.
	} \label{tbl:results}
	\centering
	\begin{tabular}{rcrrcccc}
		                Model &                    Parameters                     & $\cardinality{\States}$ & $\cardinality{\BsccsOp}$ & \texttt{PRISM} & \texttt{Classic} & \texttt{Naive} & \texttt{Sample} \\
		\midrule
		         \textbf{brp} &                \texttt{N=64,MAX=5}                &                 5{,}192 &                      134 &      1.2       &        11        &       TO       &       4.9       \\
		        \textbf{nand} &                 \texttt{N=15,K=2}                 &                56{,}128 &                       16 &      4.9       &        30        &       TO       &       64        \\
		\textbf{zeroconf\_dl} & \texttt{\tiny reset=false,deadline=40,N=1000,K=1} &               251{,}740 &                 10{,}048 &       99       &       238        &      8.0       &       1.0       \\
		       \textbf{phil4} &                                                   &                 9{,}440 &                        1 &  \texttt{err}  &        TO        &      \multicolumn{2}{c}{51}      \\
		      \textbf{rabin3} &                                                   &                27{,}766 &                        1 &  \texttt{err}  &        MO        &     \multicolumn{2}{c}{178}      \\
		\midrule
		      \textbf{branch} &                                                   &           1{,}087{,}079 &                  1{,}000 &      155       &        TO        &       TO       &       20
	\end{tabular}
\end{table}

\paragraph{Results}
We summarize our results in \cref{tbl:results}.
We observe several points.
First, we see that the naive sampling approach can hardly handle non-trivial models.
Second, our guided sampling approach achieves significant improvements on several models over both the classical, correct method as well as the potentially unsound approach of PRISM, in particular when hardly reachable portions of the state space can be completely discarded.
However, on other models, the classical approach seems to be more appropriate, in particular on models with many likely to be reached BSCCs.
Here, the sampling approach struggles to propagate the reachability bounds of all BSCCs simultaneously.
Finally, as suggested by the \textbf{phil} and \textbf{rabin} models, using mean payoff based approximation can significantly outperform classical equation solving.
%Quite surprisingly, the classic equation system based approach also is competitive on several models.
%Moreover, PRISM's default configuration fails to converge on several models, potentially due to implementation issues.
In summary, \texttt{PRISM}, \texttt{Classic}, and \texttt{Sample} all can be the fastest method, depending on the structure of the model.
However, recall that \texttt{PRISM}'s method does not give guarantees on the result.
%We again mention that \cref{tbl:results} only is a hand-picked summary of our evaluation, see \cref{app:results} for the complete set.

\paragraph{Further Discussion}
As expected, we observed that the runtime of approximation can increase drastically for smaller precision requirements (e.g.\ $\varepsilon = 10^{-8}$) and solving the equation system precisely may actually be faster for some BSCCs.
However, especially in the combined approach, if we already have some upper bounds on the reachability probability of a certain BSCC, we do not need to solve it with the original precision.
Hence, a future version of the implementation could dynamically decide whether to solve a BSCC based on mean payoff approximation or equation solving, combining advantages of both worlds.

Secondly, this also highlights an interesting trade-off implicit to our approach:
The algorithm needs to balance between exploring unknown areas and refining bounds on known BSCCs, in particular, since exploring a new BSCC adds noticeable effort:
One more target for which the reachability has to be determined.
Here, more sophisticated heuristics could be useful.

Finally, for models with large BSCCs, such as \textbf{rabin}, we also observed that the classical linear equation approach indeed runs out of memory while a variant of the approximation algorithm can still solve it, as indicated by \cref{rem:space}.
Thus, the implementation could moreover take memory constraints into account, deciding to apply the memory-saving approach in appropriate cases.
\section{Conclusion}

We presented a new perspective on computing the stationary distribution in Markov chains by rephrasing the problem in terms of mean payoff and reachability.
We combined several recent advances for these problems to obtain a sophisticated partial-exploration based algorithm.
Our evaluation shows that on several models our new approach is significantly more performant.
As a major technical contribution, we provided a general algorithmic framework, which encompasses both the classical solution approach as well as our new method.

As hinted by the discussion above, our framework is quite flexible.
For future work, we particularly want to identify better guidance heuristics.
Specifically, based on experimental data, we conjecture that the reachability part can be improved significantly.
Moreover, due to the flexibility of our framework, we can apply different methods for each BSCC to obtain the reachability and stationary distribution.
Thus, we want to find meta-heuristics which suggest the most appropriate method in each case.
For example, for smaller BSCCs, we could use the classical, precise solution method to obtain the stationary distribution, while for larger ones we employ our mean payoff approach, and, in the spirit of \cref{rem:space}, for even larger ones we approximate them to the required precision immediately, saving memory.
Additionally, we could identify BSCCs that satisfy the conditions of specialized approaches such as \cite{DBLP:conf/asmta/FourneauQ12}.

\vfill

\bibliographystyle{splncs04}
\bibliography{main}

\iftoggle{arxiv}{
\clearpage
\appendix
\section{Further Details}

\subsection{Periodicity and Limiting vs.\ Stationary} \label{app:periodic}

\begin{figure}[t]
	\centering
	\begin{tikzpicture}[auto,initial text=,xscale=2]
		\node[state] at (0, 0) (s) {$s_1$};
		\node[state] at (1, 0) (t) {$s_2$};

		\path[probedge]
			(s) edge[bend left] node[prob] {$1.0$} (t)
			(t) edge[bend left] node[prob] {$1.0$} (s)
		;
	\end{tikzpicture}
	\caption{
		Example MC to demonstrate periodicity.
	} \label{fig:periodic}
\end{figure}
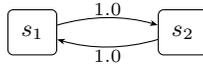

For a Markov chain $\MC$, let $\delta^n(s, s') = \ProbabilityMC<\MC, s>[\stepreach<n> \{s'\}]$ the probability to go from $s$ to $s'$ in exactly $n$ steps.
The period $p(s)$ of a state $s$ is the greatest common divisor of all $n > 0$ with $\delta^n(s, s) > 0$.
The state $s$ is called \emph{periodic} if $p(s) > 0$ and \emph{aperiodic} otherwise.
A Markov chain is aperiodic if all of its states are aperiodic.
Intuitively, in a periodic chain, many classical limits do not converge and one would instead need to consider, e.g., the Ceasaro limit.
For example, consider the MC in \cref{fig:periodic}, which is also discussed in \cite[Sec.~A.4]{DBLP:books/wi/Puterman94}.
Both states have a period of $2$ and thus the MC is periodic.
We have that $P = \begin{psmallmatrix} 0 & 1 \\ 1 & 0 \end{psmallmatrix}$.
Clearly, $P^{2n} = I$ and $P^{2n+1} = P$, and in particular $\lim_{n \to \infty} P^n$ does not exist.
In this case, the \emph{limiting distribution} starting from some initial state $\initialstate$ given as $\lim_{n \to \infty} \pi_{\MC, \initialstate}^n(s)$ is not defined, whereas the \emph{stationary distribution} as given by \cref{def:stationary} equals $(\frac{1}{2}, \frac{1}{2})$.
Citing Feller (taken from \cite[p.~337]{DBLP:books/wi/Puterman94})
\begin{quote}
	The modifications required for periodic chains are rather trite, but the formulations required become unpleasantly involved.
\end{quote}
Fortunately, in our case we can w.l.o.g.\ assume that Markov chains are aperiodic by applying a simple transformation (see e.g.\ \cite[Sect.~8.5.4]{DBLP:books/wi/Puterman94} or \cite{DBLP:conf/cav/AshokCDKM17}).
Intuitively, we simply add a self-loop with probability $0 < \alpha < 1$ to every state and rescale all remaining transitions by $(1 - \alpha)$.
In other words, we obtain a new transition matrix $P_\alpha = \alpha I + (1 - \alpha) P$.
By determining the stationary distribution $\pi_\alpha$ for this transformed Markov chain, we can directly obtain the original stationary distributions as follows:
Recall that $\pi_\alpha = P_\alpha \pi_\alpha = \alpha \pi_\alpha + (1 - \alpha) P \pi_\alpha$.
Rearranging yields that $(1 - \alpha) \pi_\alpha = (1 - \alpha) P \pi_\alpha$, i.e.\ $\pi_\alpha$ is the stationary distribution of the original, periodic Markov chain, however obtained by dealing with an aperiodic chain.
In particular, since the modified Markov chain is aperiodic, the power method also converges.

%See \cite[Sec.~8.3]{DBLP:books/wi/Puterman94} for an extensive classification scheme of Markov processes.

\subsection{Value Iteration} \label{app:value_iteration}

In this section, we provide some further insights in value iteration and some of its classical uses.
For the sake of illustration, let us consider reachability.
%We define $S_0$ as the set of states which cannot reach the target $\Reachset$, i.e.\ all states with $\ProbabilityMC<\MC, s>[\reach \Reachset] = 0$.
%This set can be found by simple graph analysis.
%Then, $\ProbabilityMC<\MC, \initialstate>[\reach \Reachset]$ is the unique solution of \cite[Thm.~10.19]{DBLP:books/daglib/0020348}
%%
%\begin{equation*}
%	f(s) = 1 \text{ if $s \in \Reachset$, } \quad 0 \text{ if $s \in S_0$, and } \quad \ExpectedSumMC{\mctransitions}{s}{f} \text{ otherwise}.
%\end{equation*}
%%
As mentioned in \cref{eq:value_iteration}, VI starts from an initial value vector $v_1[s] = 0$, and we apply the iteration
\begin{equation*}
	v_{k+1}[s] = 1 \text{ if $s \in \Reachset$, } \quad 0 \text{ if $s \in S_0$, and } \quad \ExpectedSumMC{\mctransitions}{s}{v_k} \text{ otherwise}.
\end{equation*}
This iteration is monotone and converges to the true value in the limit from below \cite[Thm.~10.15]{DBLP:books/daglib/0020348}, \cite[Thm.~7.2.12]{DBLP:books/wi/Puterman94}.
%Note that in this case, i.e.\ (unbounded) reachability, VI intrinsically can only yield approximate solutions.
There exist MC where convergence up to a given precision takes exponential time \cite[Thm.~3]{DBLP:journals/tcs/HaddadM18}, but in practice VI often is much faster than methods based on equation solving.

We are still missing an important ingredient for a practical implementation of VI.
So far, we only know that \emph{eventually} the value vector $v_k$ is close to the optimum, but we do not have a concrete (practical) bound.
For (unbounded) reachability, the a-priori bound is exponential, yet typically convergence only takes a much smaller number of steps.
Thus, we want to know immediately when the currently computed values actually are close to the true value, allowing to stop the iteration early.
This idea is formalized by a so called \emph{stopping criterion}, a method to decide whether the computation has converged.
Surprisingly, even for reachability such a stopping criterion was not known until a few years ago and all model checking implementations resorted to a best-effort solution without any actual guarantees \cite[Sec.~3.1]{DBLP:journals/tcs/HaddadM18}, \cite{DBLP:conf/cav/Baier0L0W17}.
In \cite{DBLP:conf/atva/BrazdilCCFKKPU14,DBLP:journals/tcs/HaddadM18}, a stopping criterion for reachability was independently discovered by additionally computing converging \emph{upper} bounds.\footnote{
	These works focus on Markov decision processes, i.e.\ MC with non-determinism.
	There, computing upper bounds is more involved, however for MC we can simply start with $v_1(s) = 1$ for all $s \notin S_0$ and apply the iteration from \cref{eq:value_iteration}.
}
The difference between upper and lower bounds then gives a straightforward stopping criterion:
Once the difference between upper and lower bound in the initial state is smaller than $\varepsilon$, we can stop the iteration.

A big advantage of VI is its simplicity and extendability.
For example, the iteration for reachability can be applied \emph{asynchronously}.
Here, we do not update the values of all states simultaneously.
Instead, we apply the operator of \cref{eq:value_iteration} to a subset of states.
We can thus focus the computational effort on important areas of the system instead of applying the iteration globally.
Convergence guarantees can be obtained by a simple fairness constraint.

Together, these features of VI make it a perfect candidate for heuristic-based techniques:
We repeatedly apply the value iteration operation on heuristically selected regions until the stopping criterion is satisfied.

\subsection{Mean Payoff} \label{app:mean_payoff}

We briefly outline how mean payoff can be determined using value iteration.
Recall that we iteratively compute the \emph{expected total reward}, i.e.\ the expected sum of rewards we obtain after $n$ steps.
Formally, for $v_1 = \vec{0} \in \Reals^{\cardinality{\States}}$, we iterate $v_{n+1}(s) = \reward(s) + \sum_{s' \in \States} \mctransitions(s, s') \cdot v_n(s') = \reward(s) + \ExpectedSumMC{\mctransitions}{s}{v_n}$.
In case the Markov chain is \emph{aperiodic} (which we assume w.l.o.g., see \cref{app:periodic}), the \emph{increase} $\Delta_n(s) = v_{n+1}(s) - v_n(s)$ approximates the mean payoff, i.e.\ $\meanpayoff_\reward(s) = \lim_{n \to \infty} \Delta_n(s)$ \cite[Thm.~9.4.5 a)]{DBLP:books/wi/Puterman94}.
So, intuitively, the mean payoff is the reward we can expect to obtain on average in one step after running for a long time.
(Indeed, $\frac{1}{n} v_n$ also approximates the mean payoff, even on periodic chains, however we do not have a stopping criterion for this view.)

Surprisingly, this increase even directly yields correct bounds on the mean payoff:
We have $\min_{s' \in S} \Delta_n(s') \leq \meanpayoff_\reward(s) \leq \max_{s' \in S} \Delta_n(s')$, and thus a stopping criterion \cite[Thm.~9.4.5 b)]{DBLP:books/wi/Puterman94}.
These bounds naturally will converge if two states have different mean payoff, i.e.\ $\meanpayoff_\reward(s) \neq \meanpayoff_\reward(s')$ for any $s, s' \in \States$.
The converse holds, too: If all states have the same mean payoff, these bounds do eventually converge (again in the case of aperiodicity), yielding a complete approximation scheme \cite[Cor.~9.4.6]{DBLP:books/wi/Puterman94}.
Whenever the set of states $\States$ is a single BSCC, all states have the same value.

We highlight that for the mean payoff computation the choice of the initial vector $v_1$ actually is arbitrary, hence the computation can be (i)~paused and restarted at any time by simply taking the old iteration value as \enquote{new} $v_1$ and (ii)~initialized with heuristic guesses.
As with VI for reachability, this iteration may take an exponential number of steps to reach a precision of $\varepsilon$ in the worst case, but typically is much faster (see e.g.\ \cite{DBLP:conf/cav/AshokCDKM17}).

\section{Proofs}

\subsection{Proof of \cref{stm:framework_correct}} \label{app:proof:framework_correct}

\begin{proof}
	In order to prove correctness, we show
		(I)~that $l^{\reach R}_n(s) \leq \ProbabilityMC<\MC, s>[\reach R] \leq u^{\reach R}_n(s)$ for every state $s$, step $n$, and BSCC $R \in \mathcal{B}_n$,
		(II)~that $l^R_n(s) \leq \steadystate<R>(s) \leq u^R_n(s)$ for every step $n$, explored BSCC $R \in \mathcal{B}_n$ and state $s \in R$, and, using these claims,
		(III)~that the returned bounds $(l, u)$ are correct and $\varepsilon$-precise.
	Note that we do not prove termination here, instead the main goal is to show that the stopping condition implies $\varepsilon$-correctness of the result.

	We prove claims (I) and (II) directly by checking each update of all bound-variables.
	We first consider the reachability bounds $l^{\reach R}_n$ and $u^{\reach R}_n$.
	The bounds are initialized to trivial values in \cref{line:framework:init}.
	For states in BSCCs, they are updated in \cref{line:framework:bscc_explore_reach_lower_bounds,line:framework:bscc_explore_reach_upper_bounds}, which is correct due to \cref{stm:mc_bscc_properties}.
	Finally, the updates in \cref{line:framework:bssc_update_bounds} are correct by assumption.
	Similarly, for the bounds on the stationary distribution $l^R_n$ and $u^R_n$, the claim immediately follows from the initialization in \cref{line:framework:init} and the correct-by-assumption update in \cref{line:framework:bssc_update_bounds}.

%	Before we proceed with (III), we prove two auxiliary statements.
%	i.e.\ the probability to reach a BSCC not in $\mathcal{B}_n$ is at most $\varepsilon$.
%	Second, $1 - {\sum}_{R \in \mathcal{B}_n} l^{\reach R}_n(\initialstate) \leq \varepsilon$ also implies that $L = {\sum}_{R \in \mathcal{B}_n} l^{\reach R}_n(\initialstate) \geq 1 - \varepsilon$.
%	Consequently $1 - L + l^{\reach R}_n(\initialstate) \leq l^{\reach R}_n(\initialstate) + \varepsilon$ is a correct upper bound of $\ProbabilityMC<\MC, \initialstate>[\reach R]$ for all $R \in \mathcal{B}_n$.
	%Note that together with (I) and (II), this proves the correctness claim.

	For claim (III), we show that for every state $s$ we have $u(s) - l(s) \leq \varepsilon$ and $l(s) \leq \steadystate<\MC, \initialstate>(s) \leq u(s)$.
	Thus, fix an arbitrary $s \in \States$.
	If $s$ is not in any BSCC, the claim holds directly: As \Call{UpdateBSCCs}{} always yields correct BSCCs by assumption, we surely have that $s \in \States \setminus B_n$, and thus $l(s) = u(s) = 0$ by \cref{line:framework:result_zero_states}.
	Now, if $s$ is in a BSCC $R \in \Bsccs(\MC)$, we distinguish two cases.
	If $R \notin \mathcal{B}_n$, we also have that $l(s) = u(s) = 0$ by \cref{line:framework:result_zero_states}.
	We show that $\ProbabilityMC<\MC, \initialstate>[\reach R] \leq \varepsilon$, and thus, by \cref{eq:steady_state_decomposition}, $\steadystate<\MC, \initialstate>(s) \leq \varepsilon$.
	Recall that $\sum_{R \in \Bsccs(\MC)} \ProbabilityMC<\MC, \initialstate>[\reach R] = 1$ by \cref{stm:mc_bscc_properties}.
	Since furthermore the reachability bounds are correct by (I), observe that
	\begin{equation*}
		\ProbabilityMC<\MC, \initialstate>[\reach R] \leq {\sum}_{R' \in \Bsccs(\MC) \setminus \mathcal{B}_n} \ProbabilityMC<\MC, \initialstate>[\reach R'] \leq (1 - {\sum}_{R' \in \mathcal{B}_n} l^{\reach R'}_n(\initialstate)) \leq \varepsilon,
	\end{equation*}
	For the other case of $R \in \mathcal{B}_n$, observe that, since $1 - {\sum}_{R' \in \mathcal{B}_n} l^{\reach R'}_n(\initialstate) \leq \varepsilon$, we have $L = {\sum}_{R' \in \mathcal{B}_n} l^{\reach R'}_n(\initialstate) \geq 1 - \varepsilon$.
	Consequently $1 - L + l^{\reach R}_n(\initialstate) \leq l^{\reach R}_n(\initialstate) + \varepsilon$ is a correct upper bound of $\ProbabilityMC<\MC, \initialstate>[\reach R]$.
	Together, we obtain
	\begin{align*}
		u(s) - l(s) & = \min(u^{\reach R}_n(\initialstate), l^{\reach R}_n(\initialstate) + 1 - L) \cdot u^R_n(s) - l^{\reach R}_n(\initialstate) \cdot l^R_n(s) \\
		            & \leq (l^{\reach R}_n(\initialstate) + 1 - L) \cdot u^R_n(s) - l^{\reach R}_n(\initialstate) \cdot l^R_n(s)                                 \\
		            & = l^{\reach R}_n(\initialstate) \cdot (u^R_n(s) - l^R_n(s)) + (1 - L) \cdot u^R_n(s)                                                       \\
		            & \leq l^{\reach R}_n(\initialstate) \cdot (u^R_n(s) - l^R_n(s)) + 1 - L                                                                     \\
		            & \leq l^{\reach R}_n(\initialstate) \cdot {\max}_{s' \in R} (u^R_n(s') - l^R_n(s')) + 1 - L                                                 \\
		            & \leq \varepsilon,
	\end{align*}
	where the last inequality directly follows from the termination condition of \cref{line:framework:loop} and the definition of $L$ in \cref{line:framework:result_lower_bounds_sum}.
	This concludes the proof. \qed
\end{proof}

\subsection{General Termination Criterion} \label{app:proof:framework_termination}
\begin{theorem} \label{stm:framework_termination}
	\Cref{alg:framework} terminates if additionally to the assumptions of \cref{stm:framework_correct} we have that all calls to all template functions terminate and
	\begin{enumerate}[label=(T.\roman*),labelsep=*,leftmargin=*]
		\item \label{stm:framwork_termination:bsccs}
		\Call{UpdateBSCCs}{} eventually identifies all relevant BSCCs, i.e.\ eventually $\sum_{R \in \mathcal{B}_n} \ProbabilityMC<\MC, \initialstate>[\reach R] \geq 1 - \frac{\varepsilon}{4}$,
		\item \label{stm:framwork_termination:select_dist}
		\Call{SelectDistributionUpdates}{} repeatedly selects each explored BSCC $R \in \mathcal{B}_n$ with $u^{\reach R}_n(\initialstate) \cdot \max_{s \in R} (u^R_n(s) - l^R_n(s)) \geq \frac{\varepsilon}{4 (\cardinality{\mathcal{B}_n} + 1)}$,
		\item \label{stm:framwork_termination:refine_dist}
		for each repeatedly selected BSCC $R$, \Call{RefineDistribution}{} yields arbitrarily precise solutions in the limit,
		\item \label{stm:framwork_termination:select_reach}
		\Call{SelectReachUpdates}{} repeatedly selects each explored BSCC $R \in \mathcal{B}_n$ with $u^{\reach R}(\initialstate) - l^{\reach R}(\initialstate) \geq \frac{\varepsilon}{4 (\cardinality{\mathcal{B}_n} + 1)}$, and
		\item \label{stm:framwork_termination:update_reach}
		for every repeatedly selected BSCC $R$, \Call{RefineReach}{} yields arbitrarily precise solutions in the limit for $\initialstate$.
	\end{enumerate}
	The algorithm terminates with probability 1 if all assumptions hold almost surely.\footnote{This is relevant for sampling based approaches: In general, these only explore the whole state-space with probability 1.}
\end{theorem}
\begin{remark}
	In contrast to the proofs of e.g.\ \cite{DBLP:journals/lmcs/KretinskyM20}, we do not simply assume that, for example, \enquote{eventually all BSCCs are explored and all BSCCs are selected}, since this would restrict the algorithm.
	With the current formulation, we can completely stop exploring once enough BSCCs have been discovered and still obtain termination guarantees.
\end{remark}
\begin{proof}
	Suppose the algorithm does not terminate.
	Since all individual calls to the functions terminate by assumptions, this means that the loop condition of \cref{line:framework:loop} is always satisfied.
	We derive a contradiction.

	First, as we assume that the set $\mathcal{B}_n$ increases monotonically and $\Bsccs(\MC)$ is finite, this set necessarily eventually stabilizes.
	Let $\mathcal{B}$ denote this \enquote{stable} set and $n_0$ the first step with $\mathcal{B}_{n_0} = \mathcal{B}$.
	Due to \ref{stm:framwork_termination:bsccs}, $\mathcal{B}$ contains all relevant BSCCs.
	In the following, we assume that all $n \geq n_0$ and thus $\mathcal{B}_n = \mathcal{B}$.

	We prove that both parts of the loop criterion eventually are smaller than $\frac{\varepsilon}{2}$, contradicting the assumption.
	For the first part, \ref{stm:framwork_termination:select_reach} together with \ref{stm:framwork_termination:update_reach} yields that for each BSCC in $R \in \mathcal{B}$, the reachability bounds eventually are $\varepsilon^\reach = \frac{\varepsilon}{4 (\cardinality{\mathcal{B}_n} + 1)}$ precise in $\initialstate$:
	Assume that the bounds for some BSCC $R \in \mathcal{B}_n$ are more than $\varepsilon^\reach$ apart.
	Then, \ref{stm:framwork_termination:select_reach} is applicable and $R$ is selected infinitely often.
	By \ref{stm:framwork_termination:update_reach}, the bounds eventually are arbitrarily close to the correct value, thus there exists a step $n$ where they are $\varepsilon^\reach$ precise.
	Consequently,
	\begin{align*}
		{\sum}_{R \in \mathcal{B}_n} l^{\reach R}_n(\initialstate) & \geq {\sum}_{R \in \mathcal{B}_n} (\ProbabilityMC<\MC, \initialstate>[\reach R] - \varepsilon^\reach)                              \\
		                                                           & > {\sum}_{R \in \mathcal{B}_n} \ProbabilityMC<\MC, \initialstate>[\reach R] - \cardinality{\mathcal{B}_n} \cdot \varepsilon^\reach \\
		                                                           & \geq 1 - \frac{\varepsilon}{4} - \frac{\varepsilon}{4} \geq 1 - \frac{\varepsilon}{2},
	\end{align*}
	where the penultimate step follows from \ref{stm:framwork_termination:bsccs}, showing the first part.
	The second part follows from \ref{stm:framwork_termination:select_dist} and \ref{stm:framwork_termination:update_reach}, observing that $u^{\reach R}_n(s) \leq 1$.

	The proof for almost sure termination is analogous. \qed
\end{proof}

\subsection{Proof of \cref{stm:sampling}} \label{app:proof:sampling}

\begin{proof}
	For correctness, observe that our choice of functions directly satisfy the assumptions of \cref{stm:framework_correct}, in particular the refinement approaches yield correct and monotone results as explained in the previous section.

	To show that \cref{stm:framework_termination} is applicable, we prove all conditions separately.
	First, observe that all functions terminate, in particular sampling terminates due to pigeon-hole principle:
	After sampling long enough, we eventually have to see states repeatedly, since $\cardinality{\States} < \infty$.

	Since the system eventually ends up in a BSCC with probability 1, the set $P$ necessarily has to contain some set of BSCCs $\mathcal{B} \subseteq \Bsccs(\MC)$ with $\ProbabilityMC<\MC, \initialstate>[\reach \Union_{R \in \mathcal{B}} R] \geq 1 - \frac{\varepsilon}{4}$ due to \ref{stm:sampling:core}.
	As \Call{UpdateBSSCs}{} eventually finds all BSCCs in $P$, we eventually have $\mathcal{B}_n = \mathcal{B}$, satisfying \ref{stm:framwork_termination:bsccs}.

	Next, observe that since \Call{RefineReach}{} is back-propagating values, we always have that $u^{\reach R}_n(s) \leq \ExpectedSumMC{\mctransitions}{s}{u^{\reach R}_n}$ and dually $l^{\reach R}_n(s) \geq \ExpectedSumMC{\mctransitions}{s}{u^{\reach R}_n}$:
	The only way upper bounds can decrease in a state $s \notin B_n$ (or lower bounds increase) is by updating with the expected sum.
	Consequently, for any state $s$ there necessarily has to exist at least one successor of $s'$ with $E_n^B(s) \leq E_n^B(s')$.
	Thus, if for any BSCC $R$ the condition \ref{stm:framwork_termination:select_dist} is applicable, there is a path from $\initialstate$ to $R$ with $E^B_n(s) \geq \frac{\varepsilon}{4 (\cardinality{B_n} + 1)}$ for all states $s$ along this path, and, by \ref{stm:sampling:initial} and \ref{stm:sampling:successor}, we have that the BSCC $R$ is sampled (and thus updated) infinitely often, proving \ref{stm:framwork_termination:select_dist}.
	Next, \ref{stm:framwork_termination:refine_dist} (convergence of distribution refinement) follows from the observations on \cref{alg:approx_bscc}.
	Condition~\ref{stm:framwork_termination:select_reach} follows by an analogous argument:
	If the imprecision in any BSCC is too large, there exists a path from the initial state for which \ref{stm:sampling:successor} is applicable.
	Finally, \ref{stm:framwork_termination:update_reach} follows from the fact that back-propagation of lower and upper bounds converges, see e.g.\ \cite{DBLP:journals/tcs/HaddadM18}, or, more directly, observe that the reachability value is the unique fixed point of \cref{eq:reachability_value_fixpoint}. \qed
\end{proof}

\section{Complete Instantiation} \label{app:algorithm}

We discuss our choice of sampling heuristic and present a complete instantiation of \cref{alg:framework}.
\subsection{Challenges for Sampling Heuristics}
Recall that in each iteration, we want to sample a set of states and base all functions on this sampled set.
We have several options to obtain samples.
We start at a state $s$ (not necessarily the initial state) and repeatedly select a successor, obtaining a path in the system.
We continue this process until we (i)~encounter a known BSCC or (ii)~sample already sampled states too often (this happens when we encounter an unexplored BSCC).

A naive choice would be to sample according to the transition dynamics of the Markov chain, i.e.\ select a successor based on $\mctransitions(s)$.
This eventually samples all states, and every state is sampled arbitrarily often, directly satisfying the conditions of \cref{stm:sampling}.
However, we aim for a more sophisticated approach.
Indeed, as reported in \cite{DBLP:conf/atva/BrazdilCCFKKPU14} and confirmed by subsequent works, guiding the sampling by a heuristic increases the performance drastically.
The idea is to steer the sampling towards regions which (i)~have a high influence on the result and (ii)~currently are not solved precisely.
Since these previous works compute a single value (e.g.\ reachability), they used the difference between lower and upper bounds as guidance.
In other words, they sample a successor proportional to $\mctransitions(s, s') \cdot e(s')$, where $e(s')$ denotes the error bound.
This steers the computation towards regions which both are reasonably likely to be reached (due to weighing with the transition probability) as well as with a significant uncertainty about the result (due to weighing with $e$).
As such, focussing computation on these areas should improve the result quickly.

Unfortunately, the situation is not as simple in our case, since we have several bounds to worry about:
For each BSCC, we need to compute its distribution and its reachability, thus instead of a single bound, we have 2 for each found BSCC.
Moreover, for each BSCC, the error bound of the distribution as well as the reachability may converge at completely different rates.
So, even if one of them is solved precisely, we still have to focus some effort on this BSCC.
In particular, considering e.g.\ $(u^{\reach R}_n(s) - l^{\reach R}_n(s)) \cdot \max_{s' \in R} (u^R_n(s) - l^R_n(s))$ is wrong, since we may be able to solve the stationary distribution of $R$ precisely long before the reachability bounds have converged -- the weight would be zero even though there is still computation to be performed.
However, considering $(u^{\reach R}_n(s) - l^{\reach R}_n(s)) + \max_{s' \in R} (u^R_n(s) - l^R_n(s))$, i.e.\ adding the errors instead of multiplying, does not reflect our goals:
Even BSCCs which provably are hardly reached and thus have little impact on the result would be selected very often.
As such, designing an efficient guidance heuristic seems to be significantly more intricate than in previous works.
In particular, note that the heuristic should additionally be easy to evaluate, otherwise we waste more time on computing the heuristic than we save in the end by focussing on the right areas.

%However, we aim for more sophisticated methods, and present a guided approach.
%
%First, similar to the previous methods, we propose to guide the sampling proportional to a single error bound.
%Naively, we could simply compute $E_n(s')$ from \cref{stm:sampling} \ref{stm:sampling:successor} at every step for every successor.
%However, computing $E_n$ scales linearly with the number of discovered BSCCs.
%Since some models have several thousand BSCCs, this is not viable on a larger scale.
%Instead, we aim to find a simple value which approximates $E_n(s')$.
%Note that we do not necessarily need to bound $E_n(s')$ from above or below with our value; we only require that the error bound is non-zero whenever \ref{stm:sampling:successor} is applicable.
%Essentially, there are two sources for error:
%Precision of the distribution in each BSCC and the reachability bounds.
%As such, we back-propagate these errors during \Call{RefineReach}{}.
%First, we simply propagate the distribution error $\errorbound_n^R$, i.e.\ for each sampled state $s$ we set $\errorbound_{n+1}(s) = \ExpectedSumMC{\mctransitions}{s}{\errorbound_n}$ (where $\errorbound_n(s) = \errorbound_n^R$ for $s \in R$).
%For the reachability bounds, we could not find a simple aggregate, thus we resort to caching the value of $\max_{R \in \mathcal{B}_n} (u^{\reach R}_n(s) - l^{\reach R}_n(s))$ and update it lazily.

\subsection{Our Sampling Mechanism}
%
%Our second approach works differently:
%Instead of trying to find regions of general interest, i.e.\ where $E_n(s)$ is large, we first fix an \enquote{interesting} BSCC as goal and then steer the sampling towards it.
%Recall that we essentially try to tackle two bounds per BSCC, namely on its reachability and its stationary distribution.
We propose to randomly select such a \enquote{target} before each sampling run, which gives us a fixed guidance heuristic.
%In particular, for each BSCC, we select the \enquote{reach} and \enquote{distribution} goals proportional to $\errorbound^{\reach R}_n(\initialstate)$ and $u^{\reach R}_n(\initialstate) \cdot \errorbound_n^R$, respectively.
%If the \enquote{reach} goal is selected for BSCC $R$, we sample states proportional to $\mctransitions(s, s') \cdot \errorbound_n^{\reach R}(s')$, i.e.\ weighted by the reachability error (inspired by \cite{DBLP:conf/atva/BrazdilCCFKKPU14}).
%For the \enquote{distribution} target, we instead sample proportional to $\mctransitions(s, s') \cdot u^{\reach R}_n(s')$, i.e.\ weighted with the reachability upper bound.
%In other words, we eagerly try to reach $R$.
In particular, we select a BSCC proportionally to $\errorbound^{\reach R}_n(\initialstate) + u^{\reach R}_n(\initialstate) \cdot \errorbound_n^R$, i.e.\ the sum of reachability error and distribution error weighted by the upper bound on reachability.
Note that $\errorbound^{\reach R}_n(\initialstate) + \errorbound_n^R = 0$ only if we precisely determined both the probability of reaching $R$ from $\initialstate$ as well as the stationary distribution of $R$.
Otherwise, there always is the chance of selecting $R$ and reaching it. %in the subsequent sampling run.
Hence, \cref{stm:sampling}~\ref{stm:sampling:successor} is applicable.

This does not necessarily lead us to explore new BSCCs, i.e.\ \ref{stm:sampling:core} is not satisfied.
Thus, we also explore \enquote{the unknown}, which is another type of sampling target.
Here, we make use of the approach of \cite{DBLP:journals/lmcs/KretinskyM20}, which we summarize briefly.
We introduce the bound $\errorbound^{\reach ?}_n(s)$, which represents an upper bound on the probability to reach an unexplored state.
Initially, we set $\errorbound^{\reach ?}_1(s) = 1$ for all states.
Whenever we explore a BSCC, we set $\errorbound^{\reach ?}_n(s) = 0$ for all states in the BSCC.
For all other states, we simply back-propagate $\errorbound^{\reach ?}_{n+1}(s) = \ExpectedSumMC{\mctransitions}{s}{\errorbound^{\reach ?}_n}$.
We select the \enquote{explore} goal with probability proportional to $\errorbound^{\reach ?}_n(\initialstate)$ and sample proportional to $\mctransitions(s, s') \cdot \errorbound^{\reach ?}_n(s')$.
From \cite[Sec.~3.2]{DBLP:journals/lmcs/KretinskyM20}, we know that $\errorbound^\reach_n$ is an upper bound on the probability to reach unexplored states and sampling weighted by $\errorbound^{\reach ?}_n$ eventually leads us to explore the whole system.
In particular, $\errorbound^\reach_n(\initialstate) = 0$ implies that all states (and BSCCs) have been discovered, i.e.\ $P = \States$ and $\mathcal{B}_n = \Bsccs(\MC)$, satisfying \ref{stm:sampling:core}.

\subsection{Complete Instantiation}
For readability, we omit step indices of the variables and assume that all variables are initialized to sensible values on their first read.
We also omit some technicalities that arise if, for example, $\initialstate$ is in a BSCC.

\begin{algorithm}[h]
	\caption{Concrete Instantiation of our Framework}
	\begin{algorithmic}[1]
		\Require Markov chain $\MC = (\States, \mctransitions)$, initial state $\initialstate$, precision $\varepsilon > 0$
		\Ensure $\varepsilon$-precise bounds $l, u$ on the stationary distribution $\steadystate<\MC, \initialstate>$

		\State $P \gets \{\initialstate\}$
		\While {$\big( 1 - \sum_{R \in \mathcal{B}} l^{\reach R}(\initialstate) \big) + \sum_{R \in \mathcal{B}} \big( l^{\reach R}(\initialstate) \cdot \errorbound^R \big) > \varepsilon$}
			\State $T \gets \Call{SelectTarget}{}$ \Comment{Select the guidance target}
			\If {$T = \text{reach unknown}$} $f \gets s' \mapsto \delta(s, s') \cdot \errorbound^{\reach ?}(s')$
			\ElsIf {$T = \text{reach $R$}$} $f \gets s' \mapsto \delta(s, s') \cdot u^{\reach R}(s')$
			% \ElsIf {$T = \text{update $R$}$} $f \gets s' \mapsto \delta(s, s') \cdot u^{\reach R}(s')$
			\EndIf
			\State $\rho \gets ()$, $s \gets \initialstate$
			\While {$s \notin B$ and $s \notin \rho$} \Comment{Sample path $\rho$ through $\MC$}
				\State $s \gets \Call{SampleWeighted}{s' \mapsto \delta(s, s') \cdot f(s')}$
				\If {$s$ is \texttt{null}} \textbf{break} \Comment{For example, if all successors have $f(s') = 0$}
				\EndIf
				\State $\rho \gets \rho \circ s$
			\EndWhile
			\State $P \gets P \union \rho$ \Comment{Mark all states of $\rho$ as explored}
			\If {$s \in R$ for some $R \in \mathcal{B}$} \Comment{Update reached BSCC if any}
				\State Update $l^R$ and $u^R$ according to \cref{alg:approx_bscc} (or other method)
			\Else
				\State $\mathcal{B} \gets \Call{UpdateBSSCs}{P}$, $B \gets \Union_{R \in \mathcal{B}} R$ \Comment{Discover new BSCCs}
				\For {newly discovered $R$, $s \in R$} \Comment{Update trivial bounds}
					\State $l^{\reach R}(s) \gets 1$
					\For {$\circ \neq R$} \Comment{Probability to reach others is zero (even unexplored ones)}
						\State $u^{\reach \circ}(s) \gets 0$
					\EndFor
					\State $\errorbound^{\reach ?}(s) \gets 0$
					\If {$\cardinality{R} = 1$} $l^R(s) \gets 1$
					\EndIf
				\EndFor
			\EndIf
			\For {$s \in \rho$ in reverse order} \Comment{Update reachability}
				\State $l^{\reach R}(s) \gets \ExpectedSumMC{\mctransitions}{s}{l^{\reach R}}$
				\State $u^{\reach R}(s) \gets \ExpectedSumMC{\mctransitions}{s}{u^{\reach R}}$
				\State $\errorbound^{\reach ?}(s) \gets \ExpectedSumMC{\mctransitions}{s}{\errorbound^{\reach ?}}$
			\EndFor
		\EndWhile
		\State \Return $(l, u)$ as described in \cref{alg:framework}
	\end{algorithmic}
\end{algorithm}

Observe that updating the reachability values in backwards order does not immediately instantiate \cref{alg:framework}:
There, all values of the sampled states are updated simultaneously.
However, recall that \cref{alg:framework} does not require that the set of sampled states is a path or that it contains $\initialstate$.
Thus, we can instantiate \cref{alg:framework} by sampling a path first and then select one state at a time in reverse order to update.

\pagebreak

\section{Complete Evaluation} \label{app:results}

In this section, we give all results of our evaluation.
We use the following notation:
Each table begins with the name of the model and the set constants, if any.
We then report the number of states in the model and, in the case of models with multiple BSCCs, the number of these.
Then, we list for each applicable tool the time until termination in seconds.
We write \enquote{T/O} for a timeout, \enquote{M/O} for memout, \enquote{SO} for stack overflow error (PRISM's implementation of Tarjan's algorithm is recursive and fails on larger models), and \enquote{Conv} for a convergence error (i.e.\ when PRISM did not converge in $10^{7}$ iterations).
We ordered the models by size and when a tool fails we skip evaluation for the larger ones, denoted by \enquote{-}.
For readability, we omit all models where all tools are skipped.

The used tools are:
\begin{itemize}
	\item \texttt{Sample}: guided sampling and mean payoff based approximation in BSCCs,
	\item \texttt{Naive}: unguided sampling and mean payoff based approximation in BSCCs,
	\item \texttt{Solve}: guided sampling and linear equation solving in BSCCs,
	\item \texttt{PRISM-H}: PRISM with hybrid engine, and
	\item \texttt{PRISM-E}: PRISM with explicit engine.
\end{itemize}
Note that on models comprising a single BSCC, the methods \texttt{Sample} and \texttt{Naive} coincide.
Moreover, on MDP models, \texttt{PRISM-H} is not applicable for technical reasons (we cannot specify a uniform strategy easily).
Recall that \texttt{PRISM}'s results are not guaranteed to be correct, thus \texttt{Classic} should be used as baseline for comparison of our methods.

We first give a quick overview of the results in \cref{fig:scatter}.
We observe that both \texttt{Classic} and \texttt{Sample} can be magnitudes faster on both cases.

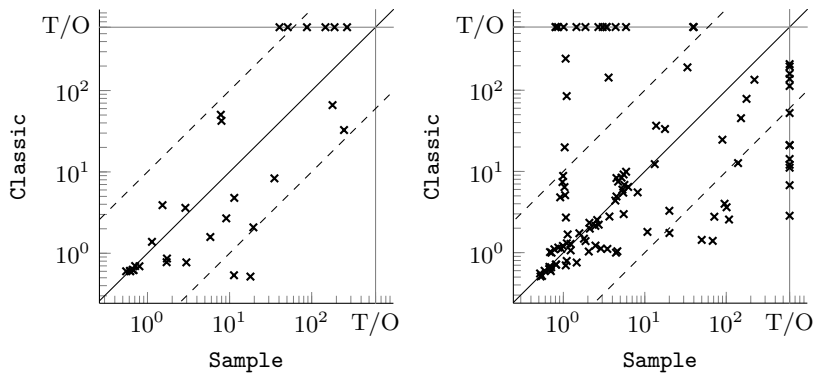
\begin{figure}[t]
	\centering
	\begin{tikzpicture}[auto]
	\begin{axis}[
			width=0.45\textwidth,height=0.45\textwidth,
			table/col sep=comma,
			xlabel=\texttt{Sample},ylabel=\texttt{Classic},
			extra x ticks={600},extra x tick labels={T/O},
			extra y ticks={600},extra y tick labels={T/O},
			xmin=0,ymin=0,xmax=990,ymax=990,xmode=log,ymode=log,
			axis x line*=bottom,
			axis y line*=left
		]
		\addplot[black,forget plot,update limits=false] coordinates {(0.0000001,0.0000001) (1000,1000)};
		\addplot[gray,forget plot,update limits=false] coordinates {(0.0000001,600) (1000,600)};
		\addplot[gray,forget plot,update limits=false] coordinates {(600,0.0000001) (600,1000)};
		\addplot[black,dashed,forget plot,update limits=false] coordinates {(0.1,1) (1000,10000)};
		\addplot[black,dashed,forget plot,update limits=false] coordinates {(1,0.1) (10000,1000)};

		\addplot+[only marks,draw=black,mark=x,thick]  table [x=approx, y=solve] {data/single.csv};
	\end{axis}
	\end{tikzpicture} %
	\begin{tikzpicture}[auto]
	\begin{axis}[
			width=0.45\textwidth,height=0.45\textwidth,
			table/col sep=comma,
			xlabel=\texttt{Sample},ylabel=\texttt{Classic},
			extra x ticks={600},extra x tick labels={T/O},
			extra y ticks={600},extra y tick labels={T/O},
			xmin=0,ymin=0,xmax=990,ymax=990,xmode=log,ymode=log,
			axis x line*=bottom,
			axis y line*=left
		]
		\addplot[black,forget plot,update limits=false] coordinates {(0.0000001,0.0000001) (1000,1000)};
		\addplot[gray,forget plot,update limits=false] coordinates {(0.0000001,600) (1000,600)};
		\addplot[gray,forget plot,update limits=false] coordinates {(600,0.0000001) (600,1000)};
		\addplot[black,dashed,forget plot,update limits=false] coordinates {(0.1,1) (1000,10000)};
		\addplot[black,dashed,forget plot,update limits=false] coordinates {(1,0.1) (10000,1000)};
		\addplot+[only marks,draw=black,mark=x,thick] table[x=approx,y=solve] {data/onestate.csv};
	\end{axis}
	\end{tikzpicture}
	\caption{
		Scatter plot comparing \texttt{Sample} and \texttt{Classic} on single BSCC models (left) and one-state BSCC models (right).
		We plot all models for which at least one method produced a result and count timeouts as 600 seconds (twice the timeout value).
		Note that the plot is logarithmic.
		The dashed lines denote a 10x difference.
	} \label{fig:scatter}
\end{figure}

\begin{table}[!t]
	\caption{
		DTMC and CTMC models comprising a single SCC.
	}
	\centering
	\small
	\begin{tabular}{lr@{\hspace{4ex}}rr@{\hspace{4ex}}rr}
		Name \& Constants &    Size & \texttt{Sample} & \texttt{Classic} & \texttt{PRISM-H} & \texttt{PRISM-E} \\
		\midrule
		cluster/N=2       &     276 &            2.98 &             0.77 &             Conv &             0.71 \\
		cluster/N=4       &     820 &           19.50 &             2.08 &                - &             1.02 \\
		cluster/N=8       &    2772 &          179.50 &            66.20 &                - &             1.09 \\
		cluster/N=16      &   10132 &             T/O &              T/O &                - &               SO \\
		cluster/N=32      &   38676 &               - &                - &                - &                - \\
		cluster/N=64      &  151060 &               - &                - &                - &                - \\
		cluster/N=128     &  597012 &               - &                - &                - &                - \\
		cluster/N=256     & 2373652 &               - &                - &                - &                - \\
		\midrule
		fms/n=1           &      54 &            0.62 &             0.60 &             0.66 &             0.71 \\
		fms/n=2           &     810 &            9.11 &             2.69 &             0.88 &             0.82 \\
		fms/n=3           &    6520 &          269.79 &              T/O &             1.42 &               SO \\
		fms/n=4           &   35910 &             T/O &                - &             3.59 &                - \\
		fms/n=5           &  152712 &               - &                - &            10.77 &                - \\
		fms/n=6           &  537768 &               - &                - &            26.78 &                - \\
		fms/n=7           & 1639440 &               - &                - &            76.40 &                - \\
		fms/n=8           & 4459455 &               - &                - &           207.90 &                - \\
		\midrule
		kanban/t=1        &     160 &            0.71 &             0.70 &             Conv &             0.71 \\
		kanban/t=2        &    4600 &           40.59 &              T/O &                - &             1.26 \\
		kanban/t=3        &   58400 &             T/O &                - &                - &               SO \\
		kanban/t=4        &  454475 &               - &                - &                - &                - \\
		kanban/t=5        & 2546432 &               - &                - &                - &                - \\
		\midrule
		mapk\_cascade/N=1 &     118 &            1.71 &             0.78 &             Conv &             0.80 \\
		mapk\_cascade/N=2 &    2172 &          247.69 &            32.71 &                - &             1.25 \\
		mapk\_cascade/N=3 &   18292 &             T/O &              M/O &                - &               SO \\
		mapk\_cascade/N=4 &   99535 &               - &                - &                - &                - \\
		mapk\_cascade/N=5 &  408366 &               - &                - &                - &                - \\
		mapk\_cascade/N=6 & 1373026 &               - &                - &                - &                - \\
		mapk\_cascade/N=7 & 3979348 &               - &                - &                - &                - \\
		\midrule
		poll3             &      36 &            0.66 &             0.63 &             0.66 &             0.77 \\
		poll4             &      96 &            0.80 &             0.69 &             0.69 &             0.72 \\
		poll5             &     240 &            1.73 &             0.86 &             0.69 &             0.79 \\
		poll6             &     576 &            5.83 &             1.58 &             0.79 &             0.83 \\
		poll7             &    1344 &           35.17 &             8.33 &             0.80 &             1.07 \\
		poll8             &    3072 &          147.11 &              T/O &             0.98 &             1.33 \\
		poll9             &    6912 &             T/O &                - &             1.20 &             1.74 \\
		poll10            &   15360 &               - &                - &             1.68 &               SO \\
		\midrule
		tandem/c=5        &      66 &            0.61 &             0.61 &             0.67 &             0.69 \\
		tandem/c=7        &     120 &            0.68 &             0.62 &             0.62 &             0.76 \\
		tandem/c=15       &     496 &            1.13 &             1.38 &             0.70 &             0.89 \\
		tandem/c=31       &    2016 &            7.97 &            42.43 &             0.81 &             1.01 \\
		tandem/c=63       &    8128 &          189.98 &              T/O &             1.09 &             1.62 \\
		tandem/c=127      &   32640 &             T/O &                - &             1.81 &               SO \\
		tandem/c=255      &  130816 &               - &                - &             4.20 &                - \\
		tandem/c=511      &  523776 &               - &                - &            16.30 &                - \\
		tandem/c=1023     & 2096128 &               - &                - &           113.48 &                - \\
		tandem/c=2047     & 8386560 &               - &                - &              T/O &                -
	\end{tabular}
\end{table}

\begin{table}
	\caption{
		MDP models where the induced DTMC only comprises a single SCC.
	}
	\centering
	\small

	\begin{tabular}{lr@{\hspace{4ex}}rr@{\hspace{4ex}}r}
		Name \& Constants &   Size & \texttt{Sample} & \texttt{Classic} & \texttt{PRISM-E} \\
		\midrule
		mutual3           &   2368 &            7.82 &            50.40 &             0.98 \\
		mutual4           &  27600 &             T/O &              M/O &               SO \\
		\midrule
		phil3             &    956 &            1.52 &             3.90 &             0.90 \\
		phil4             &   9440 &           51.08 &              T/O &               SO \\
		phil5             &  93068 &             T/O &                - &                - \\
		phil-nofair3      &    956 &            2.89 &             3.62 &             0.81 \\
		\midrule
		phil-nofair4      &   9440 &           87.60 &              T/O &               SO \\
		phil-nofair5      &  93068 &             T/O &                - &                - \\
		\midrule
		rabin3            &  27766 &          180.43 &              M/O &               SO \\
		rabin4            & 668836 &             T/O &                - &                -
	\end{tabular}
\end{table}

\begin{table}
	\caption{
		DTMC and CTMC models comprising only single-state BSCCs.
		For crowds, we abbreviate the constants TotalRuns and CrowdSize by TR and CS for readability.
	}
	\centering
	\scriptsize
	\begin{tabular}{lrr@{\hspace{4ex}}rrrr@{\hspace{4ex}}rr}
		Name \& Constants     &      Size & Components & \texttt{Naive} & \texttt{Solve} & \texttt{Sample} & \texttt{Classic} & \texttt{PRISM-H} & \texttt{PRISM-E} \\
		\midrule
		embedded/MAX\_COUNT=2 &      3478 &         36 &          63.12 &           6.01 &            4.35 &             4.38 &             1.46 &             1.74 \\
		embedded/MAX\_COUNT=3 &      4323 &         36 &         113.72 &           6.39 &            4.49 &             4.91 &             1.93 &             1.87 \\
		embedded/MAX\_COUNT=4 &      5168 &         36 &          71.18 &           7.17 &            5.31 &             5.97 &             1.88 &             1.93 \\
		embedded/MAX\_COUNT=5 &      6013 &         36 &          77.03 &           7.73 &            5.61 &             6.83 &             2.00 &             2.10 \\
		embedded/MAX\_COUNT=6 &      6858 &         36 &          81.70 &           6.60 &            4.71 &             7.49 &             1.85 &             2.20 \\
		embedded/MAX\_COUNT=7 &      7703 &         36 &          91.53 &           7.71 &            5.18 &             8.39 &             2.41 &             2.46 \\
		embedded/MAX\_COUNT=8 &      8548 &         36 &          87.77 &           7.88 &            5.49 &             9.18 &             2.36 &             2.65 \\
		\midrule
		brp/N=16,MAX=2        &       677 &         35 &          35.74 &           1.12 &            0.95 &             1.12 &             1.24 &             0.87 \\
		brp/N=16,MAX=3        &       886 &         36 &           8.70 &           1.27 &            0.99 &             1.21 &             1.28 &             0.82 \\
		brp/N=16,MAX=4        &      1095 &         37 &           8.75 &           1.28 &            1.10 &             1.30 &             1.34 &             0.83 \\
		brp/N=16,MAX=5        &      1304 &         38 &           8.16 &           1.58 &            1.13 &             1.68 &             1.51 &             0.91 \\
		brp/N=32,MAX=2        &      1349 &         67 &            T/O &           2.81 &            2.47 &             2.14 &             2.15 &             0.97 \\
		brp/N=32,MAX=3        &      1766 &         68 &              - &           2.89 &            2.11 &             1.98 &             2.33 &             1.03 \\
		brp/N=32,MAX=4        &      2183 &         69 &              - &           3.04 &            2.09 &             2.33 &             2.45 &             1.10 \\
		brp/N=32,MAX=5        &      2600 &         70 &              - &           4.18 &            2.60 &             2.52 &             2.62 &             1.18 \\
		brp/N=64,MAX=2        &      2693 &        131 &              - &           5.61 &            5.44 &             5.48 &             5.76 &             1.30 \\
		brp/N=64,MAX=3        &      3526 &        132 &              - &           7.65 &            6.25 &             6.44 &             6.15 &             1.32 \\
		brp/N=64,MAX=4        &      4359 &        133 &              - &           6.68 &            4.49 &             8.29 &             6.60 &             1.49 \\
		brp/N=64,MAX=5        &      5192 &        134 &              - &           7.39 &            5.89 &             9.91 &             6.75 &             1.49 \\
		\midrule
		crowds/TR=3,CS=5      &      1198 &         56 &            T/O &           3.38 &            3.47 &             1.11 &             1.40 &             1.01 \\
		crowds/TR=4,CS=5      &      3515 &        126 &              - &          16.98 &           10.78 &             1.80 &             3.31 &             1.14 \\
		crowds/TR=3,CS=10     &      6563 &        286 &              - &          29.90 &           20.02 &             3.28 &             5.87 &             1.71 \\
		crowds/TR=5,CS=5      &      8653 &        252 &              - &         169.08 &          101.00 &             3.62 &             8.49 &             1.83 \\
		crowds/TR=6,CS=5      &     18817 &        462 &              - &            T/O &             T/O &             6.77 &            22.48 &             2.87 \\
		crowds/TR=3,CS=15     &     19228 &        816 &              - &              - &               - &            11.12 &            24.22 &             4.39 \\
		crowds/TR=4,CS=10     &     30070 &       1001 &              - &              - &               - &            20.99 &            44.19 &             6.22 \\
		crowds/TR=3,CS=20     &     42318 &       1771 &              - &              - &               - &            52.40 &            82.78 &            11.73 \\
		crowds/TR=5,CS=10     &    111294 &       3003 &              - &              - &               - &           138.22 &           299.29 &            62.68 \\
		crowds/TR=4,CS=15     &    119800 &       3876 &              - &              - &               - &           209.68 &              T/O &            91.18 \\
		crowds/TR=4,CS=20     &    333455 &      10626 &              - &              - &               - &              T/O &                - &              T/O \\
		\midrule
		egl/N=5,L=2           &     33790 &          1 &           2.43 &           1.79 &            1.79 &             1.49 &             1.07 &             1.71 \\
		egl/N=5,L=4           &     74750 &          1 &           4.06 &           2.79 &            2.72 &             2.21 &             1.63 &             2.33 \\
		egl/N=5,L=6           &    115710 &          1 &           5.34 &           4.02 &            3.68 &             2.79 &             2.12 &             2.84 \\
		egl/N=5,L=8           &    156670 &          1 &           6.28 &           4.64 &            5.51 &             2.98 &             3.20 &             2.96 \\
		egl/N=10,L=2          &  66060286 &          1 &            M/O &            T/O &             M/O &              T/O &           136.59 &              M/O \\
		egl/N=10,L=4          & 149946366 &          1 &              - &              - &               - &                - &              T/O &                - \\
		\midrule
		leader\_sync3\_2      &        26 &          1 &            T/O &           0.53 &            0.52 &             0.55 &             0.65 &             0.66 \\
		leader\_sync4\_2      &        61 &          1 &              - &           0.58 &            0.55 &             0.52 &             0.62 &             0.76 \\
		leader\_sync3\_3      &        69 &          1 &              - &           0.62 &            0.52 &             0.51 &             0.60 &             0.68 \\
		leader\_sync5\_2      &       141 &          1 &              - &           0.67 &            0.70 &             0.59 &             0.66 &             0.62 \\
		leader\_sync3\_4      &       147 &          1 &              - &           0.59 &            0.59 &             0.60 &             0.61 &             0.63 \\
		leader\_sync4\_3      &       274 &          1 &              - &           0.67 &            0.68 &             0.62 &             0.70 &             0.63 \\
		leader\_sync4\_4      &       812 &          1 &              - &           0.72 &            0.71 &             0.66 &             0.69 &             0.79 \\
		leader\_sync5\_3      &      1050 &          1 &              - &           0.80 &            0.81 &             0.71 &             0.73 &             0.78 \\
		leader\_sync5\_4       &      4244 &          1 &              - &           1.40 &            1.10 &             0.79 &             0.97 &             0.82 \\
		\midrule
		nand/N=20,K=1         &     78332 &         21 &            T/O &         114.10 &           17.70 &            33.31 &             8.80 &             3.80 \\
		nand/N=20,K=2         &    154942 &         21 &              - &            T/O &          221.86 &           134.69 &            18.65 &             9.15 \\
		nand/N=20,K=3         &    231552 &         21 &              - &              - &             T/O &           112.79 &            31.86 &            16.96 \\
		nand/N=20,K=4         &    308162 &         21 &              - &              - &               - &           196.03 &            48.30 &            28.14 \\
		nand/N=40,K=1         &   1004862 &         41 &              - &              - &               - &              T/O &           204.29 &            96.39 \\
		nand/N=40,K=2         &   2003082 &         41 &              - &              - &               - &                - &              T/O &              T/O
	\end{tabular}
\end{table}

\begin{table}[!t]
	\caption{
		MDP models where the induced DTMC only comprises single-state BSCCs.
		All zeroconf\_dl models also have constants N=1000,K=1, which we omitted for readability.
	}
	\centering
	\scriptsize
	\begin{tabular}{lr@{\hspace{4ex}}rrrrr@{\hspace{4ex}}r}
		Name \& Constants                       &     Size & BSCCs & \texttt{Naive} & \texttt{Solve} & \texttt{Sample} & \texttt{Classic} & \texttt{PRISM-E} \\
		\midrule
		coin2/K=2                               &      272 &     8 &            T/O &           1.41 &            1.23 &             1.07 &             1.12 \\
		coin2/K=4                               &      528 &     8 &              - &          95.60 &           50.02 &             1.44 &             1.18 \\
		coin2/K=8                               &     1040 &     8 &              - &            T/O &             T/O &             2.85 &             1.52 \\
		coin2/K=16                              &     2064 &     8 &              - &              - &               - &            14.23 &             2.73 \\
		coin4/K=2                               &    22656 &    64 &              - &              - &               - &           159.56 &            33.29 \\
		coin4/K=4                               &    43136 &    64 &              - &              - &               - &              T/O &               SO \\
		\midrule
		csma2\_2                                &     1038 &     3 &            T/O &           0.98 &            1.07 &             0.70 &             0.81 \\
		csma2\_4                                &     7958 &     7 &              - &           4.75 &            4.60 &             1.04 &             1.10 \\
		csma3\_2                                &    36850 &     7 &              - &          18.04 &           20.08 &             1.74 &             1.76 \\
		csma2\_6                                &    66718 &    27 &              - &            T/O &           71.37 &             2.77 &             2.84 \\
		csma4\_2                                &   761962 &     9 &              - &              - &             T/O &            12.09 &            10.42 \\
		csma3\_4                                &  1460287 &    13 &              - &              - &               - &            21.13 &            15.83 \\
		csma3\_6                                & 84856004 &   125 &              - &              - &               - &              M/O &              M/O \\
		\midrule
		firewire/delay=3                        &     4093 &     2 &          13.22 &           1.65 &            1.57 &             1.73 &             1.35 \\
		firewire/delay=36                       &   212268 &     2 &            T/O &            T/O &           13.79 &            36.58 &               SO \\
		firewire\_abst/delay=3                  &      611 &     1 &           2.50 &           0.72 &            0.67 &             0.67 &             0.77 \\
		firewire\_abst/delay=36                 &      776 &     1 &           3.24 &          10.30 &            0.67 &             0.63 &             0.80 \\
		firewire\_dl/deadline=200,delay=3       &    14824 &   190 &            T/O &          24.47 &           13.25 &            12.31 &             6.11 \\
		firewire\_dl/deadline=200,delay=36      &    68056 &   328 &              - &          17.86 &            3.60 &           143.27 &            39.59 \\
		firewire\_dl/deadline=400,delay=3       &    69683 &   327 &              - &            T/O &           33.40 &           191.34 &            57.21 \\
		firewire\_dl/deadline=600,delay=3       &   168411 &   515 &              - &              - &           39.22 &              T/O &           227.28 \\
		firewire\_dl/deadline=400,delay=36      &   220565 &   528 &              - &              - &            3.39 &                - &           212.07 \\
		firewire\_dl/deadline=800,delay=3       &   290017 &   715 &              - &              - &           39.91 &                - &              T/O \\
		firewire\_dl/deadline=600,delay=36      &   375765 &   728 &              - &              - &            4.39 &                - &                - \\
		firewire\_dl/deadline=800,delay=36      &   530965 &   928 &              - &              - &            5.88 &                - &                - \\
		firewire\_impl\_dl/deadline=200,delay=3 &    80980 &  1007 &            T/O &            T/O &          177.43 &            78.35 &            23.75 \\
		firewire\_impl\_dl/deadline=400,delay=3 &   434364 &  3638 &              - &              - &             M/O &              T/O &              T/O \\
		\midrule
		leader3                                 &      364 &     3 &           1.30 &           0.82 &            0.81 &             0.71 &             0.87 \\
		leader4                                 &     3172 &     4 &           5.79 &           1.90 &            1.87 &             1.39 &             1.20 \\
		leader5                                 &    27299 &     5 &         116.45 &           7.96 &            8.17 &             5.53 &             3.53 \\
		leader6                                 &   237656 &     6 &            T/O &          95.19 &           89.66 &            24.58 &            13.96 \\
		leader7                                 &  2095783 &     7 &              - &            T/O &             T/O &              T/O &               SO \\
		\midrule
		wlan0/COL=0                             &     2954 &     1 &           2.85 &           2.55 &            1.46 &             0.76 &             0.91 \\
		wlan1/COL=0                             &     8625 &     1 &          14.59 &           2.40 &            4.43 &             1.00 &             1.07 \\
		wlan2/COL=0                             &    28480 &     1 &         247.80 &           5.11 &           68.22 &             1.40 &             1.67 \\
		wlan3/COL=0                             &    96302 &     1 &            T/O &            T/O &          108.18 &             2.56 &             2.64 \\
		wlan4/COL=0                             &   345000 &     1 &              - &              - &           95.33 &             4.00 &             4.19 \\
		wlan5/COL=0                             &  1295218 &     1 &              - &              - &          139.64 &            12.67 &            12.08 \\
		wlan6/COL=0                             &  5007548 &     1 &              - &              - &          151.80 &            45.26 &            42.98 \\
		\midrule
		wlan\_dl0/deadline=80                   &   189703 &  2940 &            T/O &            T/O &             T/O &              T/O &           204.29 \\
		wlan\_dl1/deadline=80                   &   450627 &  7855 &              - &              - &               - &                - &              T/O \\
		\midrule
		zeroconf/reset=true,N=1000,K=2          &      670 &    23 &            T/O &           0.85 &            0.72 &             1.00 &             1.14 \\
		zeroconf/reset=true,N=20,K=2            &      670 &    23 &              - &           0.71 &            0.69 &             1.02 &             1.03 \\
		zeroconf/reset=true,N=1000,K=4          &     1088 &    23 &              - &           0.89 &            0.78 &             1.10 &             1.01 \\
		zeroconf/reset=true,N=20,K=4            &     1088 &    23 &              - &           0.78 &            2.05 &             1.03 &             1.17 \\
		zeroconf/reset=true,N=1000,K=6          &     1506 &    23 &              - &           0.89 &            0.87 &             1.15 &             1.19 \\
		zeroconf/reset=true,N=20,K=6            &     1506 &    23 &              - &           0.84 &            2.89 &             1.12 &             1.30 \\
		zeroconf/reset=true,N=1000,K=8          &     1924 &    23 &              - &           0.97 &            1.22 &             1.28 &             1.22 \\
		zeroconf/reset=true,N=20,K=8            &     1924 &    23 &              - &           0.67 &            2.48 &             1.22 &             1.27 \\
		zeroconf/reset=false,N=1000,K=2         &    89586 &  3519 &              - &           0.99 &            0.86 &              T/O &               SO \\
		zeroconf/reset=false,N=20,K=2           &    89586 &  3519 &              - &           0.91 &            0.80 &                - &                - \\
		zeroconf/reset=false,N=1000,K=4         &   307768 &  7645 &              - &           0.99 &            0.81 &                - &                - \\
		zeroconf/reset=false,N=20,K=4           &   307768 &  7645 &              - &           0.81 &            1.01 &                - &                - \\
		\midrule
		zeroconf\_dl/reset=true,deadline=10     &     3835 &   245 &          16.32 &           1.91 &            1.08 &             2.71 &             1.99 \\
		zeroconf\_dl/reset=true,deadline=20     &     7670 &   485 &           6.39 &           2.67 &            0.91 &             4.80 &             2.31 \\
		zeroconf\_dl/reset=true,deadline=30     &    11605 &   725 &           8.88 &           1.99 &            1.03 &             6.41 &             2.80 \\
		zeroconf\_dl/reset=false,deadline=10    &    12240 &   274 &          14.36 &           2.41 &            1.05 &             5.10 &             2.76 \\
		zeroconf\_dl/reset=true,deadline=40     &    15640 &   965 &           7.93 &           2.20 &            0.99 &             8.80 &             3.30 \\
		zeroconf\_dl/reset=true,deadline=50     &    19775 &  1205 &           8.82 &           2.13 &            0.98 &             7.40 &             3.28 \\
		zeroconf\_dl/reset=false,deadline=20    &    53620 &  1192 &          12.44 &           2.77 &            1.04 &            19.81 &            10.21 \\
		zeroconf\_dl/reset=false,deadline=30    &   132806 &  4316 &           8.11 &           1.81 &            1.10 &            84.67 &            27.59 \\
		zeroconf\_dl/reset=false,deadline=40    &   251740 & 10048 &           8.05 &           4.34 &            1.07 &           244.92 &            87.36 \\
		zeroconf\_dl/reset=false,deadline=50    &   411031 & 18740 &           8.23 &           9.65 &            1.45 &              T/O &           200.76
	\end{tabular}
\end{table}

}{}

\end{document}

\section{Technical Appendix}

\subsection{Proofs -- Markov Chains}
\begin{proof}[Proof of \cref{stm:var_mc_exponential}]
	\underline{Upper bound}: By assumption, the probability to reach the goal is $1$ for every state of the chain.
	Through \cite[Lemma~5.1]{DBLP:journals/jacm/BrazdilKK14} we get that $\nstepnongoal<n> \leq 2 c^n$ for large $n$, where $c = \exp(-\cardinality{\States}^{-1} p_{\min}^{\cardinality{\States}})$.
	By solving $2 c^n = \threshold$ for $n$ we obtain the result.

	\underline{Lower bound}: Consider the upper part of the MDP in \cref{fig:exponential_memory}, i.e.\ all states $s_i$ and $r_i$, and state $d$, a Markov chain.
	Moreover, let $\goalset = \{d\}$.
	After $n+1$ steps, $p^n$ executions are in $d$ and the rest is back in state $s_0$.
	To have at least $\threshold$ many executions in $d$, we thus require at least $m = \log\threshold / \log(1 - p^n)$ such rounds with length $n + 1$.
	Consequently, we need at least $- \log\threshold \cdot (n + 1) \cdot p^{-n}$ steps, as $\log(1 - x) \leq -x$ for $x \in (0, 1)$.
\end{proof}
\begin{proof}[Proof sketch of \cref{stm:cvar_ssp_polynomial}]
	Let $P$ denote the transition matrix of the Markov chain.
	Then, $\nstepprob<n> = P^n \cdot e_1$ where $e_1$ is the unit vector corresponding to the initial state.
	By setting $P_0 = P$ and iterating $P_{k + 1} = P_k \cdot P_k$, we get $P_k = P^{2^k}$.
	We repeat this process until we obtain that $\nstepnongoal<{2^{k+1}}>{} < \threshold$.
	By \cref{stm:var_mc_exponential}, this requires only polynomially many steps, each of which comprises a matrix multiplication, which again amounts to polynomially many operations.
	Then, we know that $\VaR_\threshold$ lies between $2^k$ and $2^{k+1}$.
	(Note that the size of the entries in $P_k$ may grow exponentially, hence the overall time complexity is exponential if the algorithm is implemented with arbitrary precision arithmetic.)
	We repeat this process recursively by computing $P_k \cdot P_0 \cdot e_1$, $P_k \cdot P_1 \cdot e_1$ etc., similar to a binary search for VaR.
	After polynomially many steps, we obtain the precise VaR together with the distribution of the remaining executions.
	Finally, we compute $\expectedcost(s)$ for all states in polynomial time and together obtain CVaR by \cref{stm:cvar_equation}.
\end{proof}

\subsection{Proofs -- Markov Decision Processes}

\begin{proof}[Proof of \cref{stm:cvar_optimal_policy_exists}]
	For every $n$, let $\Strategies_n$ the set of policies achieving $\VaR_\threshold(\strategy) \leq n$.
	Clearly, $\Strategies_n \subseteq \Strategies_{n + 1}$.
	Moreover, $\Strategies_n$ is exactly the set of policies reaching the goal states with probability at least $1 - \threshold$ within $n$ steps---a closed set by optimality of deterministic policies \cite[Chap.~4]{DBLP:books/wi/Puterman94}.
	Next, we define $\Strategies_n' = \Strategies_n \setminus \Strategies_{n - 1}$.
	We have that $\{\strategy \in \Strategies \mid \CVaR_\threshold(\strategy) < \infty\} \subseteq \Union_{n \in \Naturals_0} \Strategies_n = \Union_{n \in \Naturals_0} \Strategies'_n$.
	Consequently, there exists an $n$ such that
	\begin{equation*}
		\CVaR_\threshold^* = {\inf}_{\strategy \in \Strategies'_n} \CVaR_\threshold(\strategy) = n + \tfrac{1}{\threshold} {\inf}_{\strategy \in \Strategies'_n} \nstepexpectedcost<n><\strategy>.
	\end{equation*}
	By definition of $\Strategies'_n$, the witness sequence $\strategy_i \subseteq \Strategies_n$ has an accumulation point in $\Strategies_n$.
\end{proof}
\begin{proof}[Proof of \cref{stm:memoryless_after_var}]
	Let $\strategy^*$ be an optimal stationary policy minimizing the expected time to reach the goal states $\goalset$, which always exists \cite[Prop.~2]{DBLP:journals/mor/BertsekasT91}.
	Define $\strategy'$ as follows:
	For the first $n$ steps, copy $\strategy$, and starting in the $n$-th step, follow $\strategy^*$.
	Clearly, $\nstepprob<n><\strategy> = \nstepprob<n><\strategy'>$ and thus $\VaR_\threshold(\strategy) = \VaR_\threshold(\strategy')$ as well as $\nstepnongoal<n><\strategy> = \nstepnongoal<n><\strategy'>$.
	Additionally, we have $\nstepexpectedcost<n><\strategy> \geq \nstepexpectedcost<n><\strategy'>$.
	Together, we get $\CVaR_\threshold(\strategy) \geq \CVaR_\threshold(\strategy')$ by \cref{stm:cvar_equation}.
\end{proof}
\begin{proof}[Proof of \cref{stm:cvar_lp_solution}]
	\underline{First part}: Fix a policy $\strategy$ with $\VaR_\threshold(\strategy) = n$ and $\CVaR_\threshold(\strategy) = C$.
	We construct an assignment to the LP's variables:
	Set $p_{s, i} = \nstepprob<i><\strategy>(s)$ and $p_{s, a, i} = \Probability[s_i = s, a_i = a \mid \initialstate, \strategy]$.
	This assignment satisfies the first three constraints.
	For the fourth constraint, observe that by $\VaR_\threshold(\strategy) = n$, we have that $\sum_{s \in \goalset} \nstepprob<n - 1><\strategy>(s) < 1 - \threshold \leq \sum_{s \in \goalset} \nstepprob<n><\strategy>(s)$.
	By \cref{stm:cvar_equation}, we have that $\CVaR_\threshold(\strategy) = C = n + \tfrac{1}{\threshold} \nstepexpectedcost<n><\strategy>$.
	Since $\nstepexpectedcost<n><\strategy> = \sum_{s \in \States} \nstepprob<n><\strategy>(s) \cdot \expectedcost(s) = \sum_{s \in \States} p_{s, n} \cdot \expectedcost(s)$, we get that $\nstepexpectedcost<n><\strategy> = \threshold \cdot (c - n)$, proving the claim.

	\underline{Second part}: We construct the policy $\strategy$ as follows.
	For the first $n$ steps, at step $i$ in state $s$, choose action $a$ with probability $p_{s, a, i}$.
	Afterwards, i.e.\ starting from step $n$, in state $s$ follow a policy achieving the optimal expected cost $\expectedcost(s)$ (note the similarity to \cref{stm:memoryless_after_var}).
	Clearly, $\nstepprob<i><\strategy>(s) = p_{s, i}$ and thus $v = \nstepexpectedcost<n><\strategy>$.
	Now, we need to distinguish two cases.
	We have that $\VaR_\threshold(\strategy) = n - 1$ if $\sum_{s \in \goalset} p_{s, n - 1} = 1 - \threshold$ and $\VaR_\threshold(\strategy) = n$ otherwise.
	In the latter case, we directly get that $\CVaR_\threshold(\strategy) = n + \frac{1}{\threshold} v$ by \cref{stm:cvar_equation}.
	In the former, observe that in step $n - 1$ a fraction of exactly $1 - \threshold$ executions has reached the goal states.
	Consequently, the remaining $\nstepnongoal<n-1><\strategy> = \threshold$ executions all need to perform at least one more step, thus $\nstepexpectedcost<n - 1><\strategy> = \threshold + \nstepexpectedcost<n><\strategy>$.
	Inserting yields $\CVaR_\threshold(\strategy) = (n - 1) + \frac{1}{\threshold} \nstepexpectedcost<n - 1><\strategy> = n + \frac{1}{\threshold} \nstepexpectedcost<n><\strategy> = n + \frac{1}{\threshold} v$.
\end{proof}
\begin{proof}[Proof of \cref{stm:var_bound_from_cvar}]
	We have $\VaR_\threshold(\strategy^*) \leq \CVaR_\threshold(\strategy^*) \leq \CVaR_\threshold(\strategy)$, where the first inequality follows from definition and the second from optimality of $\strategy^*$.
\end{proof}
\begin{proof}[Proof sketch of \cref{stm:mdp_cvar_exponential_bound}]
	Since a proper policy exists, there also exists a proper memoryless deterministic policy $\strategy^p$ \cite[Prop.~2]{DBLP:journals/mor/BertsekasT91}.
	As $\strategy^p$ is deterministic, \cref{stm:var_mc_exponential} is applicable with $p_{\min}$ being the smallest transition probability in the MDP, and $\VaR_\threshold(\strategy^p)$ is at most exponential, as is $e(s)$ by similar reasoning.
	Together, $n + \tfrac{1}{\threshold} \nstepexpectedcost<n><\strategy^p>$ for $n = \VaR_\threshold(\strategy^p)$ is of at most exponential size, proving the claim through \cref{stm:cvar_equation}.
\end{proof}
\begin{proof}[Proof of \cref{stm:cvar_from_pareto}]
	\underline{First part}: Fix $n$ and $(p, E) \in \pareto<n><s>$ together with the policy $\strategy$ and set $\threshold = 1 - p$.
	Clearly, $\VaR_\threshold(\strategy) \leq n$ since at least a fraction of $p$ executions reach the goal within $n$ steps.
	Let $\VaR_\threshold(\strategy) = n' \leq n$ and $p' \leq p$ the exact fraction of executions that reach within $n'$ steps.
	By \cref{stm:cvar_equation} we get $\CVaR_\threshold(\strategy) = n' + \frac{1}{\threshold} \nstepexpectedcost<n'><\strategy>$.
	It remains to show that $n' + \frac{1}{\threshold} \nstepexpectedcost<n'><\strategy> \leq n + \frac{1}{\threshold} E$.
	After $n'$ steps, at most $1 - p'$ executions have not reached the goal.
	Thus $\nstepexpectedcost<n'><\strategy> \leq (1 - p') \cdot (n - n') + E \leq \threshold \cdot (n - n') + E$.

	\underline{Second part}: Fix a state $s$, a policy $\strategy$ and assume that $\VaR_\threshold(\strategy) = n$ and $\CVaR_\threshold(\strategy) = C$.
	We show that $(1 - \threshold, \threshold \cdot (C - n)) \in \pareto<n><s>$.
	The probability to reach the goal states in $n$ steps under $\strategy$ is at least $1 - \threshold$ by the definition of VaR, proving the first component.
	From \cref{stm:cvar_equation} we get $\CVaR_\threshold(\strategy) = C = n + \frac{1}{\threshold} \nstepexpectedcost<n><\strategy>$.
	Reordering yields $\nstepexpectedcost<n><\strategy> = \threshold \cdot (E - n)$, proving the second component.
\end{proof}
\begin{proof}[Proof of \cref{stm:pareto_shape}]
	\underline{Closure}:
		If we have that $(p, E) \in \pareto<n><s>$, we also have that $(p', E), (p, E') \in \pareto<n><s>$ for all $0 \leq p' \leq p$ and $E' \geq E$ by definition.

	\underline{Convexity}:
		Let $\strategy$ and $\strategy'$ be two strategies corresponding to two points $(p, E), (p', E') \in \pareto<n><s>$ and fix $\lambda \in [0, 1]$.
		Following $\strategy$ with probability $\lambda$ and $\strategy'$ with probability $1 - \lambda$ reaches the goal set with at least $\lambda p + (1 - \lambda) p'$ and similar for the expectation. %; together $\lambda (p, E) + (1 - \lambda) (p', E') \in \pareto<n><s>$.

	\underline{Polygon}:
		We prove by induction that $\pareto<n><s>$ is a polygon with deterministic policies as vertices.

		For $n = 0$, we either have that $\pareto<0><s> = [0, 1] \times \Reals_{\geq 0}$ if $s \in \goalset$, or, if $s \notin \goalset$, $\pareto<0><s> = \{0\} \times [e(s), \infty)$.
		The first case follows trivially from the definition.
		For the second, observe that the probability to reach the goal in $0$ steps is zero and the expected time for all remaining executions is at least $e(s)$. %, by the definition of $e$.
		In both cases, $\pareto<0><s>$ is a polygon and the extremal points $(1, 0)$ and $(0, e(s))$, respectively, are achievable by a stationary deterministic policy.

		For the induction step, fix $n$ and a state $s$.
		We prove that $(p, E) \in \pareto<n+1><s>$ iff there exist a distribution over the actions $w : \stateactions(s) \to [0, 1]$ and achievable points $(p_{a, s'}, E_{a, s'}) \in \pareto<n><s'>$ for all $a \in \stateactions(s), s' \in \States$ such that
		\begin{equation*}
			(p, E) = {\sum}_{a \in \stateactions(s), s' \in \States} w(a) \cdot \mdptransitions(s, a, s') \cdot (p_{a, s'}, E_{a, s'})
		\end{equation*}
		The first equality follow directly from linearity of reachability:
		If the successors under action $a$ can reach the goal with probabilities $p_{a, s'}$ in $n$ steps, then the current state can reach the goal with the average of these probabilities in $n + 1$ steps (note the similarity to regular value iteration for reachability).
		For the second equality, recall the interpretation of the Pareto set:
		For all actions $a$ and successors $s'$ there exists a policy $\strategy_{a, s'}$ such that after $n$ steps at least a fraction of $p_{a, s'}$ executions have reached the goal and the expected time to reach the goal is at most $E_{a, s'}$.
		So, in state $s$, we can take one step and then simply follow these respective strategies to achieve the values in the equation.
		Dually, if there is a policy $\strategy$ for $(p, E) \in \pareto<n+1><s>$, we immediately get strategies achieving the respective values in the successors.
		Together with the induction hypothesis, this linear characterization proves that $\pareto<n+1><s>$ is a polygon, and the extremal points are achievable by Markovian deterministic policies.
\end{proof}
\begin{proof}[Proof of \cref{stm:cvar_algo_correct}]
	\underline{Termination}: We show that there always exists an $n$ such that $\pareto<n><\initialstate>$ contains $(1 - \threshold, E)$ for any $E$.
	Since we always have a proper strategy $\strategy$, $\VaR_\threshold(\strategy) = n < \infty$ and thus $(1 - \threshold, \nstepexpectedcost<n><\strategy>) \in \pareto<n><\initialstate>$ by \cref{stm:cvar_from_pareto}.
	By \cref{stm:mdp_var_exponential_bound}, we know that $\VaR_\threshold(\strategy)$ is at most exponentially large, proving the step bound.

	\underline{Correctness}: Let $\strategy^*$ be an optimal strategy, i.e.\ achieving the optimal CVaR.
	Further, let $\VaR_\threshold(\strategy^*) = n^*$ and $\CVaR_\threshold(\strategy^*) = E^*$.
	By \cref{stm:cvar_from_pareto}, we have that $(1 - \threshold, \threshold \cdot (E^* - n^*)) \in \pareto<n^*><\initialstate>$.
	Moreover, $\CVaR_\threshold(\strategy) \geq n^*$ for every strategy $\strategy$, so the algorithm runs until at least iteration $n^*$.
	Consequently, $\mathtt{c} \leq E^*$ when the algorithm terminates.
	To conclude, if we had $\mathtt{c} < E^*$, there must exist another strategy $\strategy'$ which achieves a better CVaR, again by virtue of \cref{stm:cvar_from_pareto}.
	Together, we have that $\mathtt{c} = E^*$ at the end.
\end{proof}

\bibliography{main}